\documentclass[11pt,a4paper]{article}

\usepackage{amssymb}
\usepackage{amsmath}
\usepackage{amsfonts}
\usepackage{epsfig}
\usepackage{graphicx}
\usepackage{color}
\usepackage{bigints}
\usepackage{array}
\usepackage{dsfont}

\newtheorem{definition}{Definition}[section]
\newtheorem{theorem}{Theorem}[section]
\newtheorem{lemma}{Lemma}[section]
\newtheorem{corollary}{Corollary}[section]

\newtheorem{proposition}{Proposition}[section]

\newbox\ProofSym
\setbox\ProofSym=\hbox{%
	\unitlength=0.18ex%
	\begin{picture}(10,10)
	\put(0,0){\framebox(9,9){}}
	\put(0,3){\framebox(6,6){}}
	\end{picture}}

\newcommand{\R}{\mathbb{R}}

\newcommand{\J}{\mathbb{J}}

\newcommand{\D}{\mathcal{D}}
\newcommand{\CC}{\mathbb{C}}

\newcommand{\del}{\mathrm{Del}}
\newcommand{\vor}{\mathrm{Vor}}

\newcommand{\ST}{\mathit{SplitT}}
\newcommand{\LMS}{\mathrm{LMS}}
\newcommand{\LIS}{\mathrm{LIS}}

\newcommand{\brac}[1]{\bigl[ #1 \bigr]}
\newcommand{\E}[1] {\mathrm{E}\bigl[#1\bigr]}

\newcommand{\cancel}[1] {}

\begin{document}

\title{A Generalization of Self-Improving Algorithms\thanks{Research of Cheng, Jin and Wong are supported by Research Grants Council, Hong Kong, China (project no.~16200317). Research of Chiu supported by ERC StG 757609. The conference version of this manuscript appeared in Proceeding of the 36th International Symposium on Computational Geometry, 2020, 29:1--29:13~\cite{GSI-socg-20}.}}

\author{Siu-Wing Cheng\thanks{Department of Computer Science and Engineering, HKUST, Hong Kong, China.}
  \quad \quad Man-Kwun Chiu\thanks{Institut f\"ur Informatik, Freie Universit\"at Berlin, Berlin, Germany}
  \quad \quad Kai Jin\thanks{School of Intelligent Systems Engineering, Sun Yat-Sen University, Shenzhen, China}
  \quad\quad Man Ting Wong\footnotemark[2]}

\maketitle

\begin{abstract}
Ailon~et~al.~[SICOMP'11] proposed self-improving algorithms for sorting and Delaunay triangulation (DT) when the input instances $x_1,\cdots,x_n$ follow some unknown \emph{product distribution}.  That is, $x_i$ comes from a fixed unknown distribution $\D_i$, and the $x_i$'s are drawn independently.  After spending $O(n^{1+\varepsilon})$ time in a learning phase, the subsequent expected running time is $O((n+ H)/\varepsilon)$, where $H \in \{H_\mathrm{S},H_\mathrm{DT}\}$, and $H_\mathrm{S}$ and $H_\mathrm{DT}$ are the entropies of the distributions of the sorting and DT output, respectively.  In this paper, we allow dependence among the $x_i$'s under the \emph{group product distribution}.  There is a hidden partition of $[1,n]$ into groups; the $x_i$'s in the $k$-th group are fixed unknown functions of the same hidden variable $u_k$; and the $u_k$'s are drawn from an unknown product distribution.  We describe self-improving algorithms for sorting and DT under this model when the functions that map $u_k$ to $x_i$'s are well-behaved.  After an $O(\mathrm{poly}(n))$-time training phase, we achieve $O(n + H_\mathrm{S})$ and $O(n\alpha(n) + H_\mathrm{DT})$ expected running times for sorting and DT, respectively, where $\alpha(\cdot)$ is the inverse Ackermann function.

\textbf{keywords} expected running time, entropy, sorting, Delaunay triangulation.
\end{abstract}

\maketitle

\section{Introduction}
\label{sect:introduction}

Ailon et al.~\cite{SICOMP11selfimp} proposed \emph{self-improving algorithms} for sorting and Delaunay triangulation (DT).  The setting is that the input is drawn from an unknown but fixed distribution $\D$.  The goal is to automatically compute some auxiliary structures in a \emph{training phase}, so that these structures allow an algorithm to achieve an expected running time potentially better than the worst-case optimum in the subsequent \emph{operation phase}, where the expectation is taken over the distribution $\D$.  The expected running time in the operation phase is known as the \emph{limiting complexity}.

This model is attractive for two reasons.  First, it addresses the criticism that worst-case time complexity may not be relevant because worst-case input may occur rarely, if at all.  Second, it is more general than some previous average-case analyses that deal with distributions that have simple, compact formulations such as the uniform, Poisson, and Gaussian distributions.  There is still a constraint in the work of  Ailon et al.~\cite{SICOMP11selfimp}: $\D$ must be a \emph{product distribution}, meaning that each input item follows a distribution and two distinct input items are independently drawn from their respective distributions.

Self-improving algorithms under product distributions have been proposed for sorting, DT, 2D coordinatewise maxima, and 2D convex hull.  For sorting, Ailon et al.~showed that a limiting complexity of $O((n  + H_\mathrm{S})/\varepsilon)$ can be achieved for any $\varepsilon \in (0,1)$, where $H_\mathrm{S}$ denotes the entropy of the distribution of the sorting output.  This limiting complexity is optimal in the comparison-based model by Shannon's theory~\cite{cover06}.   The training phase uses $O(n^{\varepsilon})$ input instances and runs in $O(n^{1+\varepsilon})$ time.  The probability of achieving the stated limiting complexity is at least $1 - 1/n$.  Ailon et al.~\cite{SICOMP11selfimp} also proposed a self-improving algorithm for DT.  The performance of the training phase is the same.  The limiting complexity is $O((n  + H_\mathrm{DT})/\varepsilon)$, where $H_\mathrm{DT}$ denotes the entropy of the distribution of the Delaunay triangulations.  Self-improving algorithms for 2D coordinatewise maxima and convex hulls have been developed by Clarkson et al.~\cite{clarkson14}.  The limiting complexities for the 2D maxima and 2D convex hull problems are $O(\mathrm{OptM} + n)$ and $O(\mathrm{OptC} + n\log\log n)$, respectively, where OptM and OptC are the minimum expected depths of linear decision trees for the maxima and convex hull problems.

It is natural to allow dependence among input items.  However, some restriction is necessary because Ailon et al.~showed that $\Omega(2^{n\log n})$ bits of storage are necessary for optimally sorting $n$ numbers if there is no restriction on the input distribution.  In~\cite{Algorithmica-20}, two extensions are considered for sorting.  The first extension assumes that there is a hidden partition of $[1,n]$ into groups $G_k$'s, the input numbers with indices in $G_k$ follow unknown linear functions of a common parameter $u_k$, and the parameters $u_1, u_2, \cdots$ follow a product distribution.  A limiting complexity of $O((n + H_\mathrm{S})/\varepsilon)$ can be achieved after a training phase that processes $O(n^\varepsilon)$ instances in $O(n^2\log^3 n)$ time.
The second extension assumes that the input is a hidden mixture of product distributions and an upper bound $m$ is given on the number of distributions in the mixture.  A limiting complexity of $O((n\log m + H_\mathrm{S})/\varepsilon)$ can be achieved after a training phase that processes $O(mn\log(mn))$ instances in $O(mn\log^2(mn) + m^\varepsilon n^{1+\varepsilon}\log(mn))$ time.

In this paper, we revisit the problems of sorting and DT and allow dependence among input items.  We assume that the following conditions hold:
 \begin{enumerate}
 \item[(i)] There is a hidden partition of $[1,n]$ into groups $G_1, G_2, \cdots$ such that for each $G_k$, there is a hidden parameter $u_k$ such that for all $i \in G_k$, the $i$-th input item $x_i$ is equal to $h_{i,k}(u_k)$ for some unknown function $h_{i,k}$.
     For sorting, the parameters $u_1,u_2,\ldots$ belong to $\R$ and $h_{i,k}(u_k) \in \R$. For DT, the parameters $u_1,u_2,\ldots$ belong to $\R^2$ and $h_{i,k}(u_k) \in \R^2$.
 \item[(ii)] The parameters $u_1, u_2, \ldots$ follow a product distribution.
 \end{enumerate}
 We call such an input distribution a \emph{group product distribution}.  The groups $G_k$'s are not given and they have to be learned in the training phase.
 Our generalization has the following features.
\begin{itemize}
	\item For sorting, we do not assume any specific formulation of the functions $h_{i,k}$'s.
    Neither is any oracle given for evaluating them.
    We only assume that the graph of each $h_{i,k}$ is a connected curve with at most $c_0$ extrema, where $c_0$ is a known constant,
      and the graphs of two distinct $h_{i,k}$ and $h_{j,k}$ intersect in $O(1)$ points.
    Our algorithm does not reconstruct or approximate the $h_{i,k}$'s.
	
	\item For DT, we assume that, for each $i \in G_k$, there are bivariate polynomials $h_{i,k}^x$ and $h_{i,k}^y$ in $u_k$ that give the $x$- and $y$-coordinates of the input point $x_i$, respectively.
          These polynomials have degrees no more than a fixed constant.
          No further information about these polynomials is given.
          Depending on the distribution of $u_k$, it may be impossible to reconstruct the polynomials $h_{i,k}^x$ and $h_{i,k}^y$ using the input data.
\end{itemize}

\noindent Let $\alpha(\cdot)$ be the inverse Ackermann function.  We prove that an optimal $O(n + H_\mathrm{S})$ limiting complexity for sorting and a nearly optimal $O(n\alpha(n) + H_\mathrm{DT})$ limiting complexity for DT can be achieved with probability at least $1 - O(1/n)$ after a polynomial-time training phase.  The training   takes $\tilde{O}(n^3)$ time for sorting and $\tilde{O}(n^{10})$ time for DT.

We use several new techniques to obtain our results.  To learn the hidden partition for sorting, we need to test if two indices $i$ and $j$ are in the same group.  By collecting $x_i$'s and $x_j$'s from some instances, we can reduce the test to finding the longest monotonic subsequence (LMS) among the points $(x_i,x_j)$'s.  We establish a threshold such that $i$ and $j$ are in the same group with very high probability if and only if  the LMS has length greater than or equal to the threshold.  In the operation phase, we will set up $O(n)$ ordered intervals from left to right, and we need to sort the subset of an input instance $I$ within an interval.  Under the group product distribution, there can be a large subset $I' \subseteq I$ from the same group that reside in an interval, which does not happen in the case of product distribution.  We do not have enough time to sort $I'$ from scratch.  Instead, we need to recognize that $I'$ gives a result similar to what we have seen in the training phase.  To this end, we must be able to ``read off'' the answer from some precomputed information in the training phase in order to beat the worst-case bound.  We construct a trie in the training phase which will be used in the operation phase to efficiently return an encoding of $I'$ that is equivalent to the \emph{Lehmer code}~\cite{lehmer60}.  Given such an encoding of $I'$, we can sort $I'$ in linear time.  This trie structure is essential for achieving the optimal limiting complexity.

The hidden partition for DT seems harder to learn given the 2D nature of the problem.  We employ tools from algebraic geometry to do so, which is the reason for requiring the $h_{i,k}^x$'s and $h_{i,k}^y$'s to be bivariate polynomials of fixed degree.  In the operation phase, we need to compute the Voronoi diagram of the subsets of $I$ inside the triangles of a canonical Delaunay triangulation.  Under the group product distribution, a large subset of $I$ may fall into the same triangle, which does not happen in the case of product distribution.  The challenge is to compute certain information and store it in a compact way in the training phase so that we can ``read off'' the Voronoi diagram needed.
%We need a structure that is equivalent to a Delaunay triangulation or Voronoi diagram.  Yet it should be  ``more combinatorial'' in nature so that it is not as sensitive to geometric perturbations.
The split tree fits this role nicely because the Delaunay triangulation can be computed from it in linear expected time~\cite{DTbySort,CK95}.  We expand the trie structure used for sorting to record different split trees that are generated in the training phase to facilitate the DT computation in the operation phase.

%Omitted proofs and descriptions can be found in the appendix.

\section{Self-improving sorter}

Let $\D$ denote the group product distribution of the input instances.  We use $I = (x_1,x_2,\ldots,x_n)$ to denote an input instance.  We use  $G_1, G_2, \ldots$ to denote the groups in the hidden partition of $[1,n]$.  Each group $G_k$ is governed by a hidden parameter $u_k\in \R$, and for each $i \in G_k$, there is a hidden function $h_{i,k}$ that determines $x_i = h_{i,k}(u_k)$.  We do not impose any particular formulation of the $h_{i,k}$'s as long as they satisfy the following properties:
\begin{itemize}
	\item For every $G_k$ and every $i\in G_k$, the graph of $h_{i,k}$ is a connected curve that has at most $c_0$ extrema for a known constant $c_0$.
	\item For every $G_k$ and every distinct pair $i,j\in G_k$, the graphs of $h_{i,k}$ and $h_{j,k}$ intersect in $O(1)$ points.
	\item For every $G_k$, every $i\in G_k$, and every $c\in \R$, $\Pr\bigl[h_{i,k}(u_k)=c\bigr]=0$.
\end{itemize}
Although $G_k$ is a subset of $[1,n]$, for convenience, we sometimes use ``input number in $G_k$'' to mean ``input number whose index belongs to $G_k$'' when there is no confusion.

\subsection{Hidden partition}

We first learn the hidden partition of $[1,n]$.  Given a sequence $\sigma$ of real numbers, let $\LMS(\sigma)$ be the length of the \emph{longest monotone subsequence} of $\sigma$ (either increasing or decreasing), and let $\LIS(\sigma)$ be the length of the \emph{longest increasing subsequence} of $\sigma$.

We describe how to test if the indices 1 and 2 belong to the same group.  The other index pairs can be handled in the same way.
Let $\ell =\max\bigl\{100^3,\, (90\ln(4n^3))^2,\, (6c_0+3)^2\bigr\}$, where $c_0$ is the constant upper bound on the number of extrema of $h_{i,k}$.  Take $\ell$ instances (from the group product distribution $\D$).  Let $I_1, I_2, \cdots, I_\ell$ denote these instances sorted in increasing order of their first items.  That is, $x_1^{(1)} < \cdots < x_1^{(\ell)}$, where $x_1^{(i)}$ denotes the first item in $I_i$.  Similarly, $x_2^{(i)}$ denotes the second item in $I_i$.  Then, compute $\LIS\bigl(x_2^{(1)},\ldots,x_2^{(\ell)}\bigr)$ and $\LIS\bigl(x_2^{(\ell)},\ldots,x_2^{(1)}\bigr)$ in $O(\ell\log \ell)$ time by some folklore method.  The larger of the two is $\LMS\bigl(x_2^{(1)},\ldots,x_2^{(\ell)}\bigr)$.  If $\LMS\bigl(x_2^{(1)},\ldots,x_2^{(\ell)}\bigr)\geq \ell/(2c_0+1)$, report that 1 and 2 are in the same group.  Otherwise, report that 1 and 2 are in different groups.

%It fails in testing $(x_1,x_2)$ if $x_1,x_2$ are in different groups yet $\LMS(x_2^{(1)},\ldots,x_2^{(N)})\geq N/(2\mu+1)$.
%But this occurs with probability at most $(2n^3)^{-1}$.
%Applying the union bound on $n^2/2$ pairs, the entire algorithm fails with probability at most $(4n)^{-1}$. Therefore,

We show that the above test works correctly with high probability.  The following result is obtained by applying the first theorem in~\cite{KIM-LIS} and the setting that $\ell \geq \max \{100^3,\,(90\ln(4n^3))^2\}$.

\begin{lemma}[\cite{KIM-LIS}]
	\label{lem:perm}
	If $\sigma$ is a permutation of $[1,\ell]$ drawn uniformly at random, then
        \[\Pr\Bigl[\LIS(\sigma)\geq 3\sqrt{\ell}\Bigr] \leq \exp \bigl(-\sqrt \ell/90\bigr) = O(1/n^3).\]
\end{lemma}

We are ready to show the correctness of our test procedure.

\begin{lemma}\label{lemma:learn-parititon-ST}
If the indices $1$ and $2$ belong to the same group, $\LMS\bigl(x_2^{(1)},\ldots,x_2^{(\ell)}\bigr)$ $\geq \ell/(2c_0+1)$; otherwise,
$\Pr\Bigl[\LMS\bigl(x_2^{(1)},\ldots,x_2^{(\ell)}\bigr)\geq \ell/(2c_0+1)\Bigr] = O(1/n^3)$.
\end{lemma}
\begin{proof}
	Suppose that 1 and 2 belong to group $G_k$.  Then, $x_1=h_{1,k}(u_k)$ and $x_2=h_{2,k}(u_k)$.  Let $t_1,\ldots,t_m$, where $m\leq 2c_0$, be the values of $u_k$ at the extrema of $h_{1,k}$ and $h_{2,k}$.  Let $t_0=-\infty$ and let $t_{m+1}=+\infty$.  By the pigeonhole principle, there exists $j\in [0,m]$ such that
	$\bigl|\bigl\{u_{k}^{(1)},\ldots u_{k}^{(\ell)}\bigr\}\cap [t_j,t_{j+1})\bigr| \geq \ell/(m+1)\geq \ell/(2c_0+1)$, and both $h_{1,k}$ and $h_{2,k}$ are monotonic in $[t_j,t_{j+1})$.  It follows that $\LMS\bigl(x_2^{(1)},\ldots,x_2^{(\ell)}\bigr)\geq \ell/(2c_0+1)$.
	
	Suppose that 1 and 2 belong to different groups.  The distribution of $x_2^{(1)},\cdots,x_2^{(\ell)}$ is the same as the uniform distribution of the permutations of $[1,\ell]$.  Lemma~\ref{lem:perm} implies that
  \[\Pr\bigl[\LIS\bigl(x_2^{(1)},\ldots,x_2^{(\ell)}\bigr)\geq 3\sqrt \ell\bigr] \leq O(1/n^3).\]
Symmetrically,
  \[\Pr\bigl[\LIS\bigl(x_2^{(\ell)},\ldots,x_2^{(1)}\bigr)\geq 3\sqrt \ell\bigr] \leq O(1/n^3).\]
As $\LMS\bigl(x_2^{(1)},\ldots,x_2^{(\ell)}\bigr)=\max\bigl\{\LIS\bigl(x_2^{(1)},\ldots,x_2^{(\ell)}\bigr),\LIS\bigl(x_2^{(\ell)},\ldots,x_2^{(1)}\bigr)\bigr\}$, by the union bound, we get $\Pr\bigl[\LMS\bigl(x_2^{(1)},\ldots,x_2^{(\ell)}\bigr)\geq 3\sqrt \ell\bigr]\leq O(1/n^3)$.

Since $\ell \geq (6c_0+3)^2$, we have $3\sqrt \ell \leq \ell/(2c_0+1)$. Hence,
  \[\Pr\bigl[\LMS\bigl(x_2^{(1)},\ldots,x_2^{(\ell)}\bigr)\geq \ell/(2c_0+1)\bigr] = O(1/n^3).\]
\end{proof}

There are $O(n^2)$ index pairs to check, each taking $O(\ell\log \ell)$ time.
% = O((c_0^2 + \log^2n)(\log\log n + \log c_0))$ time.
We conclude that:
\begin{corollary}
	\label{cor:sort-group}
	The partition of $\, [1,n]$ into groups can be learned in $\tilde{O}(n^2)$ time using $\tilde{O}(n^2)$ input instances.   The probability of success is at least $1 - O(1/n)$.
\end{corollary}

\subsection{The $V$-list}

Following~\cite{SICOMP11selfimp}, we define a \emph{$V$-list}, $(v_0,v_1,\cdots,v_{n+1})$, in the training phase as follows.
Take $\lambda=\lceil n^2\ln n\rceil$ instances $I_1,\ldots,I_\lambda$.
Let $y_1 < \cdots < y_{\lambda n}$ be the numbers in these instances sorted in increasing order.
Define $v_r = y_{\lambda r}$ for all $r \in [1,n]$. Moreover, define $v_0=-\infty$ and $v_{n+1}=+\infty$.
In the following, we take $[v_0,c)$ to be $(-\infty,c)$ for all $c\in \R$.  We use $I \sim \D$ to indicate that $I$ is drawn from the group product distribution $\D$.

\begin{lemma}
	\label{lemma:construct-V-sorting}
	It holds with probability at least $1-1/n^{192}$ that for every $r \in [0,n]$,
        \[\mathrm{E}_{I\sim \D}[{|I \cap [v_r,v_{r+1})|}]=O(1).\]
\end{lemma}
\begin{proof}
    Let $z_1,\ldots,z_n$ denote $x_1,\ldots,x_n$ in $I_1$, let $z_{n+1},\ldots,z_{2n}$ denote $x_1,\ldots,x_n$ in $I_2$, $\ldots$, and let  $z_{(\lambda -1)n+1},\ldots,z_{\lambda n}$ denote $x_1,\ldots,x_n$ in $I_\lambda$.  Define $z_0 = -\infty$ and $z_{\lambda n+1} = +\infty$.

    Consider any distinct pair $i,j \in [0,\lambda n+1]$.  Let $m_{ij}$ be the number of instances among $I_1,\ldots,I_\lambda$ that contain neither $z_i$ nor $z_j$.  Note that $\lambda-2\leq m_{ij} \leq \lambda$.  Denote by $I_{s_1},\cdots,I_{s_{m_{ij}}}$ the $m_{ij}$ instances that contain neither $z_i$ nor $z_j$.  Define the following random variables:
   \begin{align*}
   \forall \, a \in [1,m_{ij}], \quad Y^{(i,j)}_a  & =\left\{\begin{array}{ll}
    \bigl|I_{s_a} \cap [z_i,z_j)\bigr|, & \text{if $z_i<z_j$,} \\ [0.5em]
    \bigl|I_{s_a} \cap [z_j,z_i)\bigr|,  & \text{if $z_j<z_i$.}
    \end{array}\right. \\[1ex]
	Y^{(i,j)} & = Y^{(i,j)}_1+\ldots+Y^{(i,j)}_{m_{ij}}.
    \end{align*}
    The expectations $\E{Y^{(i,j)}_a}$ and $\E{Y^{(i,j)}}$ are taken over the distribution of $I_{s_1}, \ldots, I_{s_{m_{ij}}}$ induced by $\D$.  Given the sample $I_1,\ldots,I_\lambda$, we call~$(i,j)$ a \emph{good pair} if $\E{Y^{(i,j)}} \leq 11 \lambda \text{ or } Y^{(i,j)}>\lambda$.
    \begin{quote}
    	\begin{proposition}
    		\label{pro:good-pair}
    		With probability at least $1 - n^{-199}$, $(i,j)$ is a good pair.
    	\end{proposition}
    	\begin{proof}
			For every $a \in [1,m_{ij}]$, let $X_a=\frac{1}{n} Y^{(i,j)}_a$.  Let $X=X_1+\ldots+X_{m_{ij}}=\frac{1}{n}Y^{(i,j)}$.  Note that $X_1,\ldots,X_{m_{ij}}$ are independent and each lies in the range $[0,1]$.  By Hoeffding's inequality~\cite{Hoeffding-inequality}, for any $\beta>0$, it holds that $\Pr\bigl[\bigl|X-\mathrm{E}[X]\bigr|\geq \beta m_{ij}\bigr] < 2e^{-2m_{ij}\beta^2}$.  Setting $\beta=10 \mathrm{E}[X]/(11 m_{ij})$,
        		\[
        			\Pr\bigl[\bigl|X-\mathrm{E}[X]\bigr|\geq 10\mathrm{E}[X]/11 \bigr]< 2e^{-2m_{ij}(10 \mathrm{E}[X])^2/(11 m_{ij})^{2}}.
        		\]
    		Therefore,
    			$
    				\Pr\bigl[X<  \mathrm{E}[X]/11 \bigr] < 2e^{-200\mathrm{E}[X]^2/(121 m_{ij})}.
    			$
    		When $\mathrm{E}[X] > 11\lambda/n$,
        		\begin{align*}
        			\Pr[X< \lambda/n]  & \leq \Pr\bigl[X< \mathrm{E}[X]/11\bigr] < 2e^{-200\times 121 \lambda^2 / (121 n^2 m_{ij})} \\
        	                                       			& =2e^{-200 \lambda^2/(n^2 m_{ij})}\leq 2e^{-200 \lambda/n^2}<2n^{-200}.
        		\end{align*}
    		In other words, $\Pr\bigl[Y^{(i,j)}<\lambda \bigr] < 2n^{-200}$ when $\E{Y^{(i,j)}}>11\lambda$.  Hence, it holds with probability at least $1-n^{-199}$ that $(i,j)$ is a good pair.
    	\end{proof}
	\end{quote}
	
	There are at most $(\lambda n+2)^2 < n^{7}$ pairs of distinct indices from $[0,\lambda n+1]$.  By the union bound, it holds with probability at least $1-n^{-192}$ that all these pairs are good.  Assume in the following that all these pairs are indeed good.

    Consider two consecutive elements $v_r$ and $v_{r+1}$ in the $V$-list.  They correspond to $z_i$ and $z_j$, respectively, for some distinct $i,j \in [0,\lambda n+1]$.  By the construction of the $V$-list,  $z_i<z_j$ and the range $[z_i,z_j)$ contains fewer than $\lambda$ points among $\{z_k\mid k\neq i, k\neq j,k\in [1,\lambda n]\}$.  Therefore, fewer than $\lambda$ points from $I_{s_1},\ldots,I_{s_{m_{ij}}}$ fall in $[z_i,z_j)$, that is, $Y^{(i,j)}<\lambda$.  This implies that $\E{Y^{(i,j)}} \leq 11\lambda$ because $(i,j)$ is a good pair by assumption.

   Consider $\mathrm{E}_{I\sim \D}\bigl[|I \cap [v_r,v_{r+1})|\bigr] = \mathrm{E}_{I\sim \D}\bigl[|I \cap [z_i,z_j)|\bigr]$.  By the group product distribution model, both $\Pr \bigl[h_{a,k}(u_k) = z_i\bigr]$ and $\Pr \bigl[h_{a,k}(u_k) = z_j\bigr]$ are equal to zero for every $G_k$ and every $a \in G_k$.  This implies that $\Pr \bigl[z_i \in I\bigr]$ and $\Pr \bigl[z_j \in I\bigr]$ are equal to zero.  It follows that $\mathrm{E}_{I \sim \D}\bigl[|I \cap [z_i,z_j)|\bigr] = \E{Y^{(i,j)}_a}$ for every $a \in [1,m_{ij}]$.  As a result,
   \[
   		\mathrm{E}_{I \sim \D}\bigl[|I \cap [z_i,z_j)|\bigr] = \frac{\E{Y^{(i,j)}}}{m_{ij}} \leq \frac{\E{Y^{(i,j)}}}{\lambda - 2} \leq \frac{11\lambda}{\lambda -2} = O(1).
   	\]
\end{proof}

%Lemma~\ref{lemma:construct-V-sorting} is more general than its counterparts in~\cite{SICOMP11selfimp,Algorithmica-20} in that it does not make any assumption on the distribution $\D$.
%The distribution is allowed to be arbitrary, and it is only required that each instance is generated independently. The price to pay is that $O(n^2\log n)$ instances are needed instead of $O(\log n)$ in~\cite{SICOMP11selfimp,Algorithmica-20}.  This cost is dominated by the trie construction cost in the training phase to be discussed in the next section.

\subsection{Trie}

In the operation phase, we distribute the numbers in an instance $I$ to the intervals $[v_r,v_{r+1})$'s, sort the numbers in each interval, and concatenate the results.  We need an efficient method to distribute the numbers.  In~\cite{Algorithmica-20}, the functions $h_{i,k}$'s are linear,  so for every $k$, each $h_{i,k}$ can be reformulated as a linear function in another input number in $G_k$, say $x_1$.  These linear functions of $x_1$ together with the horizontal lines at the $v_r$'s form an arrangement of lines for $G_k$.  Cutting vertically through the arrangement vertices, we partition $\R$ into disjoint ranges of $x_1$.   The sorted order of the $x_i$'s in $G_k$ is invariant within each range of $x_1$.  Therefore, our problem reduces to finding the range that contains $x_1$.  In this paper, we cannot compute such an arrangement because we do not know the formulations of the $h_{i,k}$'s.  A different method is needed.

We define two functions below.
\begin{itemize}
	
	\item $b: \bigcup_{m=1}^n \R^m \rightarrow \bigcup_{m=1}^n [0,n]^m$ such that for every $m \in [1,n]$ and every vector $(z_1,\cdots,z_m)\in \R^m$, $b(z_1,\cdots,z_m)=(r_1,\cdots,r_m)$ such that
      for all $i \in \brac{1,m}$, $z_i \in [v_{r_i}, v_{r_i+1})$.

      \item $\pi: \bigcup_{m=1}^n \R^m \rightarrow \bigcup_{m=1}^n [0,n]^m$ such that for every $m \in [1,n]$ and every vector $(z_1,\cdots,z_m)\in \R^m$,
      $\pi(z_1,\cdots,z_m) = (j_1,\cdots, j_m)$ such that
      for all $i \in \brac{1,m}$,
      $j_i = 0$ if $z_i = \min_{a \in [1,i]} z_a$; otherwise, $j_i$ is the index of the largest element in $\{z_1,\cdots,z_{i-1}\}$ that is less than $z_i$.

\end{itemize}
For every subset $J\subseteq \brac{1,n}$, let $I|_J$ denote the subsequence of $I$ whose index set is $J$.  The output of $b(I|_{G_k})$ tells us the intervals $[v_r,v_{r+1})$'s that contain the input numbers in $I|_{G_k}$.  The output of $\pi(z_1,\ldots,z_m)$ is equivalent to the Lehmer code of $z_1,\cdots,z_{m}$~\cite{lehmer60}.  We can store $(z_1,\cdots,z_{m})$ by a doubly linked list $L$, and the output of $\pi(z_1,\ldots,z_m)$ is a list of pointers to entries in $L$.
Given $\pi(I|_{G_k})$, it is easy to sort $I|_{G_k}$ in  $O(|G_k|)$ time.

Our goal in this section is to show that we can build two tries to support efficient retrieval of  $b(I|_{G_k})$ and $\pi(I|_{G_k})$.  We will show that the expected query times of the tries are $O(|G_k|)$ plus the entropies of the distributions of $b(I|_{G_k})$ and $\pi(I|_{G_k})$, respectively.

For every subset $S \subseteq \mathbb{R}^m$, define $b(S) = \{b(z) : z \in S\}$ and $\pi(S) = \{\pi(z) : z\in S\}$.  We first bound the cardinalities of $b(S)$ and $\pi(S)$ when $S = \{I|_{G_k} : I \sim \D\}$.

\begin{lemma}
	\label{lem:sort-size}
	Let $S = \{I|_{G_k} : I \sim \D\}$.  Let $m=|G_k|$.  Then,
        \[|b(S)| = O(nm) \text{ and }|\pi(S)| = O(m^2).\]
\end{lemma}
\begin{proof}
	Without loss of generality, assume that $G_k = [1,m]$.  By the properties of the $h_{i,k}$'s, there are $O(m^2)$ intersections among the $h_{i,k}$'s for $i \in [1,m]$.  The vertical lines through these intersections divide $\R^2$  into $O(m^2)$ slabs.  The output of $\pi$ is invariant when $u_k$ is restricted to one of these slabs.  Therefore, there are $O(m^2)$ possible outcomes for $\pi(S)$, that is, $|\pi(S)| = O(m^2)$.
	
	Add horizontal lines $y=v_r$ for $r \in [1,n]$.  There are $O(nm)$ intersections between these horizontal lines and the graphs of $h_{i,k}$'s mentioned above.  The vertical lines through these intersections define $O(nm)$ slabs.  The output of $b$ is invariant when $u_k$ is restricted to any one of these slabs.  Thus, there are $O(nm)$ possible outcomes for $b(S)$, i.e., $|b(S)| = O(nm)$.
\end{proof}

Next, we prove a technical result about the entropy of the output of a function whose image is a discrete set with a polynomial size.  It will be used in analyzing the performance of the auxiliary structures that we will organize for efficient retrieval of the output of $b$ and $\pi$.

\begin{lemma}\label{lem:trie-runningtime}
Let $J$ be a subset of $\,[1,n]$.  Let $f$ be a function on $I|_J$.  Assume that there is a known value $t_0 \in \bigl[n,n^{O(1)}\bigr]$ such that $\bigl|\bigl\{ f(I|_J) : I \sim \D\bigr\}\bigr| \leq t_0$.  For every  $N \geq \lceil t_0\ln t_0 \ln n \rceil$ and every sample $I_1, \ldots, I_N$ drawn independently from $\D$, it holds with probability at least $1 - t_0^{-2}$ that
\begin{equation*}%\label{eqn:tilde-rho-sim-rho}
	\sum_{i=1}^{M} \rho_i \log (1/\tilde{\rho}_i) + \rho_0 |J| \log  n = H_J(f) + O(|J|),
\end{equation*}
where
\begin{itemize}
	\item $H_J(f)$ denotes the entropy of $f(I|_J)$ under $I \sim \D$,
	\item $\beta_1, \dots, \beta_M$ denote the distinct outcomes in $\bigl\{ f(I_s|_J) : s \in [1,N]\bigr\}$,
	\item $\tilde{\rho}_i = \frac{1}{N} \cdot$ number of occurrences of $\beta_i$ among $f(I_1|_J),\ldots,f(I_N|_J)$ for $i \in [1,M]$,
	\item $\rho_i=\Pr_{I \sim \D} \bigl[f(I|_J)=\beta_i\bigr]$ for $i \in [1,M]$, and
	\item $\rho_0 = \Pr_{I \sim \D} \bigl[f(I|_J)\notin \left\{\beta_1,\ldots,\beta_M\right\}\bigr]$.
\end{itemize}
\end{lemma}

\begin{proof}
%We will that $\sum_{i=1}^{M} \rho_i \log(1/\tilde{\rho}_i)=H_J(f) +O(1)$ and $\rho_0=O(1/\log n)$ with high probability.
%
Let $t$ be the cardinality of $\bigl\{ f(I|_J) : I \sim \D \bigr\}$.  Note that $t \leq t_0$.  Let $\beta_{M+1}$, $\beta_{M+2}$, $\ldots,$ $\beta_t$ denote the outcomes in $\{f(I|_J):I\sim \D\} \setminus \{\beta_1,\ldots,\beta_M\}$.  For every $i\in[M+1,t]$, define $\rho_i = \Pr_{I\sim \D}\bigl[f(I|_J) = \beta_i\bigr]$ and $\tilde{\rho}_i=0$.  Define $\Delta=8/(t\ln n)$.

%Now, consider any $\beta_i~(i\in [1,t])$ for which $\rho_i>\Delta$.
First, we prove a claim that for every $i \in [1,t]$, if $\rho_i > \Delta$, then $\Pr[\tilde{\rho}_i \geq \rho_i/8]  \geq 1- t_0^{-3}$.  Define $X_i = \tilde{\rho}_i N$ for $i \in [1,t]$.  We have
\[
	\E{X_i} = \rho_i N > \Delta N \geq \frac{8t_0\ln t_0 \ln n}{t\ln n} \geq 8\ln t_0.
\]
For $s \in [1,N]$, define $X_{is} = 1$ if $f(I_s|_J) = \beta_i$ and $X_{is} =0$ otherwise.  Then, $X_i = \sum_{s=1}^N X_{is}$.  Since $X_{i1},\ldots,X_{iN}$ are independent, by the Chernoff bound,
\begin{align*}
\Pr\bigl[\tilde{\rho}_i < \rho_i/8\bigr] & = \Pr\bigl[\tilde{\rho}_iN < \rho_iN/8\bigr] =
\Pr\bigl[X_i < \E{X_i}/8 \bigr] \\
& = \Pr\bigl[X_i <(1-7/8)\E{X_i}\bigr] \leq \exp\left(-\frac{1}{2}\Bigl(\frac{7}{8}\Bigr)^2\E{X_i}\right)  \\
& < \exp\left(-\frac{49 \ln t_0}{16}\right)< t_0^{-3}.
\end{align*}
This completes the proof of the claim.

By our claim and the union bound, it holds with probability at least $1-t  t_0^{-3} \geq 1-t_0^{-2}$ that
\begin{equation}\label{eqn:proof-rho-8}
  \forall \, i \in [1,t], \quad \rho_i > \Delta \Rightarrow \tilde{\rho}_i\geq \rho_i/8.
\end{equation}
Assume that \eqref{eqn:proof-rho-8} holds for the rest of this proof.  Then,
\begin{align}
	\sum_{\substack{i\in [1,M] \\ \rho_i>\Delta}} \rho_i\log(1/\tilde{\rho}_i)
     & \leq  \sum_{\substack{i \in [1,M] \\ \rho_i>\Delta}}  \rho_i\log(8/\rho_i)
	 = \sum_{\substack{i \in [1,M]\\ \rho_i>\Delta}} \rho_i\log(1/\rho_i)  + O(1) \nonumber \\
	&  \leq  H_J(f) + O(1).  \label{eqn:proof-rho-9}
\end{align}
Take $\beta_i$ for an arbitrary $i \in [1,M]$.  We have $\tilde{\rho}_i \geq 1/N$ as the frequency of $\beta_i$ is at least one.
Thus, $\log(1/\tilde{\rho}_i)\leq \log N = O(\log n)$. Therefore,
\begin{align}
	\sum_{\substack{i\in [1,M] \\ \rho_i \leq \Delta}} \rho_i \log (1/\tilde{\rho}_i)
     & \leq  \sum_{\substack{i\in [1,M] \\ \rho_i \leq \Delta}} \Delta \cdot O(\log n) \nonumber \\
     & =  \frac{8M}{t\ln n} \cdot O(\log n) = O(1).  \label{eqn:proof-rho-10}
\end{align}
For every $i \in [M+1,t]$, by \eqref{eqn:proof-rho-8} and the fact that $\tilde{\rho}_i=0$, we get $\rho_i\leq \Delta$.
Therefore,
\begin{equation}\label{eqn:proof-rho-11}
    \rho_0 |J|\log n = \left(\sum_{i=M+1}^{t} \rho_i\right) |J| \log n < t\Delta |J| \log n = \frac{8t|J|\log n}{t\ln n} = O(|J|).
\end{equation}
The lemma follows by combining \eqref{eqn:proof-rho-9}, \eqref{eqn:proof-rho-10} and \eqref{eqn:proof-rho-11}.
\end{proof}

We show how to construct two tries in the training phase, one for $b(I|_J)$ and another for $\pi(I|_J)$, so that the outputs of $b(I|_{J})$ and $\pi(I|_{J})$ can be retrieved efficiently in the operation phase.

\begin{theorem}\label{thm:algorithm-limit-output}
Let $J$ be a subset of $\,[1,n]$.    Let $f$ be $b$ or $\pi$.
    Let $N = \left\lceil t_0 \ln t_0 \ln n\right\rceil$, where $t_0=\mathrm{poly}(n)$ is a given value that bounds the size of $\,\bigl|\{ f(I|_J) : I \sim \D \}\bigr|$.  We can build a data structure such that, given $I|_J$ for an input instance $I$,  the data structure returns $f(I|_J)$ and it has the following performance guarantees.
\begin{itemize}
	\item It can be constructed in $\tilde{O}(N|J|)$ time from $N$ independent $I_1|_J, \ldots, I_N|_J$.
	\item It uses $O(t_0|J|)$ space.
	\item It holds with probability at least $1 - t_0^{-2}$ that the expected query time is $O\bigl(H_J(f) + |J|\bigr)$, where the expectation is taken over $I \sim \D$.
\end{itemize}
\end{theorem}
\begin{proof}
	Without loss of generality, assume that $J=[1,|J|]$. Hence $I|_J=(x_1,\ldots,x_{|J|})$.
	
	Let $I_1, \ldots, I_N$ denote the independent input instances that we sample for the construction of the data structure.  Define the following quantities as in Lemma~\ref{lem:trie-runningtime}:
	\begin{itemize}
		\item $\beta_1, \dots, \beta_M$ denote the distinct outcomes in $\bigl\{ f(I_s|_J): s \in [1,N]\bigr\}$, and
		\item $\tilde{\rho}_i = \frac{1}{N} \cdot$ number of occurrences of $\beta_i$ among $f(I_1|J), \ldots, f(I_N|_J)$ for $i \in [1,M]$.
		%\item $\rho_i=\Pr_{I \sim \D} \bigl[f(I|_J)=\beta_i\bigr]$ for $i \in [1,M]$, and
		%\item $\rho_0 = \Pr_{I \sim \D} \bigl[f(I|_J)\notin \left\{\beta_1,\ldots,\beta_M\right\}\bigr]$.
	\end{itemize}
	
	Our data structure is a trie $T$ for storing $\beta_1, \ldots, \beta_M$.  The edges of $T$ are given labels in $[0,n]$ and the $\beta_i$'s are stored at the leaves of $T$ such that for every $j \in [1,M]$, $\beta_j$ is equal to the string of labels on the path from the root of $T$ to the leaf that stores $\beta_j$.

Given $I|_J$ for an input instance $I$, we use $T$ to return $f(I|_J)=f(x_1,\ldots,x_{|J|})$ as follows.  Let $x_0$ be a dummy symbol to the left of $x_1$.  Start at the root of $T$ with $i=1$.  Whenever we are at $x_{i-1}$ and a node $u$, we find the child $w$ of $u$ such that the label on $uw$ corresponds to the $i$-th output element of $f(I|_J)$, and then we move to $x_i$ and $w$.  We will elaborate on the determination of $w$ shortly.  The existence of a particular child $w$ of $u$ depends on whether there is an input instance in the training phase that prompts the creation of $w$.  If we reach a leaf of $T$, we return the $\beta_j$ stored at that leaf.  If we cannot find an appropriate child to proceed at any point, we abort and compute $f(I|_J)$ from scratch in $O(|J|\log n)$ time---if $f = b$, we compute the entries in $b(I|_J)$ by binary search; if $f = \pi$, we determine $\pi(I|_J)$ by computing the sorted order of every prefix of $x_1,\ldots,x_{|J|}$ with the help of a balanced binary search tree.
	
   In the following, we elaborate on the step in the operation phase that we are at $x_{i-1}$ and a node $u$ of $T$ and we need to find the child $w$ of $u$ such that the label on $uw$ corresponds to the $i$-th output element of $f(I|_J)$.  We will also discuss what to do in the training phase in order to facilitate this step in the operation phase.
	
	\medskip
	
	\emph{Case~1.~$f = b$}: Each child $w$ of $u$ represents an interval $[v_r,v_{r+1})$ for some $r \in [0,n]$.
        So the children of $u$ represent a sorted list of disjoint intervals in $\R$, but these intervals may not cover $\R$.
	
	At the end of the training phase, the children of $u$ are stored in a nearly optimal binary search tree $A_u$ in the order of the disjoint intervals represented by them~\cite{mehlhorn75}.  The search time of $A_u$ depends on the node weights of $T$ defined as follows: the weight of a leaf that stores $\beta_j$ is $\tilde{\rho}_j$; the weight of an internal node of $T$ is the sum of the weights of all leaves descending from that internal node.    When we search for $x_i$ in $A_u$ in the operation phase, if $x_i$ is contained in the interval represented by a child $w$ of $u$, the search time is $O(\log (\text{weight}(u)/\text{weight}(w)))$.  Since $\text{weight}(u) \leq 1$ and $\text{weight}(w) \geq 1/N$, the search time is no more than $O(\log N)$ in the worst case, which implies that any unsuccessful search in $A_u$ takes $O(\log N)$ time.
	
	To build $A_u$ in the training phase, we grow a balanced binary search tree $L_u$ from being initially empty to the final set of elements in $A_u$.  There are at most $n+1$ children of $u$ as there are $n+1$ intervals $[v_r,v_{r+1})$'s for $r \in [0,n]$.   As new elements of $A_u$ are discovered, they are inserted into $L_u$ in $O(\log n)$ time each.     Also, $L_u$ provides access to the children of $u$ in $O(\log n)$ time in the training phase.  At the end of the training phase, we build $A_u$ as a nearly optimal binary search tree of the elements in $L_u$.  The construction takes time linear in the size of $A_u$~\cite{Fredman75}.
	
	The number of nodes in $T$ is $O(t_0|J|)$, and therefore, the storage required by $T$ and its auxiliary structures is $O(t_0|J|)$.  The construction of the trie $T$ and its auxiliary structures takes $O(N |J| \log n)$ time because for every $I_s$ and every input number in $I|_J$, we spend $O(\log n)$ time at some node of $T$.
	
	\vspace{4pt}
		
	\emph{Case~2. $f = \pi$}:  When we are at $x_{i-1}$ and a node $u$ of $T$, we already know the string of labels on the path from the root of $T$ to $u$.  This string of labels is equal to $\pi(x_1,\ldots,x_{i-1})$.   Also, suppose that we have inductively obtained the total order of $x_1,\ldots,x_{i-1}$ when we come to $u$.  We represent this total order of $x_1,\ldots,x_{i-1}$ by a map $\xi_u: [0,i-1] \rightarrow [0,i-1]$ such that: (i)~$x_0 = -\infty$ and $\xi_u(0) = 0$, and (ii)~for $a \in [1,i-1]$, $x_{\xi_u(a)}$ is the $a$-th smallest number among $x_1,\ldots,x_{i-1}$.
	
	%The total order $x_{\xi(1)}, \ldots, x_{\xi(i-1)}$ divides $\R$ into $i+1$ intervals which we denote by $\omega_0, \omega_1, \ldots, \omega_{i-1}$ in the left to right order.
	%The goal is to find the child $w$ of $u$ such that $\pi_i(x_1,\ldots,x_i)$ is equal to the concatenation of $\pi_{i-1}(x_1,\ldots,x_{i-1})$ and $x_{\xi(b)}$, where $\xi(b)$ is the label on the edge $uw$.

	At the end of the training phase, we store $\xi_u(0), \ldots, \xi_u(i-1)$ at the leaves of a nearly optimal binary search tree $A_u$ in the left to right order.  For each internal node $\eta$ of $A_u$, if $\eta$ separates $\xi_u(a)$ and $\xi_u(a+1)$ for some $a \in [0,i-2]$, then $\eta$ stores the index $\xi_u(a)$.  Consider a leaf $\eta'$ of $A_u$ that represents $\xi_u(a)$ for some $a \in [0,i-1]$.   If there was an instance $I_s$, $s \in [1,N]$, in the training phase such that the $i$-th entry of $\pi(I_s|_J)$ is $\xi_u(a)$, then a child $w$ of $u$ was created in the training phase such that the label of $uw$ is $\xi_u(a)$, and the leaf $\eta'$ stores a pointer to $w$.  Otherwise, the pointer at $\eta'$ is null.
	
	In the operation phase, when we search $A_u$ using $x_i$ and come to a node $\eta$, we can retrieve the index $\xi_u(a)$ stored at $\eta$ and then compare $x_i$ and $x_{\xi_u(a)}$ in $O(1)$ time.  This allows us to recursively search the left or right subtree of $\eta$ depending on whether $x_i < x_{\xi_u(a)}$ or $x_i > x_{\xi_u(a)}$, respectively.  A search in $A_u$ is aborted if we reach a node of $A_u$ such that every leaf descending from it stores a null pointer.  Eventually, either the search in $A_u$ is aborted, or we reach the leaf $\eta'$ of $A_u$ that stores the index $\xi_u(e)$ where $x_{\xi_u(e)}$ is the largest number among $x_1,\ldots,x_{i-1}$ that is less than $x_i$.  If the search in $A_u$ is aborted, we compute $\pi(I|_J)$ from scratch in $O(|J|\log n)$ time.  Otherwise, we follow the pointer stored at the leaf $\eta'$ to a child $w$ of $u$.
	
	The nodes of the trie $T$ are assigned weights as in the case of $f = b$.  If a leaf $\eta'$ of $A_u$ stores a pointer to a child $w$ of $u$, we set  the weight of $\eta'$ to be $\text{weight}(w)$; otherwise, we set the weight of $\eta'$ to be zero.  Therefore, if the search of $A_u$ using $x_i$ returns a child $w$ of $u$, then the search time is $O(\log (\text{weight}(u)/\text{weight}(w)))$.  Since $\text{weight}(u) \leq 1$ and $\text{weight}(w) \geq 1/N$, the search time is no more than $O(\log N)$ in the worst case.    Therefore, an aborted search in $A_u$ takes $O(\log N)$ time.
	
	To build $A_u$ in the training phase, we use a balanced binary search tree to keep track of the growing set of children of $u$ as in the case of $f = b$.   So we spend $O(\log n)$ time at $u$ for each instance $I_s$ in the training phase.  At the end of the training phase, $A_u$ can be built in linear time, provided that $\xi_u(a)$ is available for every $a \in [1,i-1]$.
	
	We need to construct the map $\xi_u$ in the training phase.  This is done inductively as follows.  Suppose that we have constructed the map $\xi_v: [0,i-2] \rightarrow [0,i-2]$ at the parent $v$ of $u$.  We already know the label on the edge $vu$, say $\xi_v(e)$ for some $e \in [0,i-2]$.  Therefore, if there is any input $x_1,x_2,\ldots,x_{i-1},\ldots$ that will bring us from the root of $T$ to this node $u$, the total order of $x_0,x_1,\ldots,x_{i-1}$ is equal to $x_{\xi_v(0)}, x_{\xi_v(1)},\ldots,x_{\xi_v(e)}, x_{i-1}, x_{\xi_v(e+1)}, \ldots, x_{\xi_v(i-2)}$.  Hence, we should define $\xi_u$ so that:
	
	\smallskip
	
	\begin{minipage}[t]{\textwidth}
	\begin{itemize}
		\item for $a \in [0,e]$, $\xi_u(a) = \xi_v(a)$,
		\item $\xi_u(e+1) = i-1$, and
		\item for $a \in [e+2,i-1]$, $\xi_u(a) = \xi_v(a-1)$.
	\end{itemize}
	\end{minipage}

	\smallskip

	An efficient way to set up $\xi_u$ is to represent $\xi_v$ inductively as a version of a persistent search tree~\cite{driscoll89}.  For $a \in [1,i-2]$, the $a$-th  node in the symmetric order in the persistent search tree $\xi_v$ stores $\xi_v(a)$.  Then, the construction of $\xi_u$ as described in the above can be done easily by a persistent insertion of a new node between the $e$-th and $(e+1)$-th nodes in $\xi_v$.  The new version of the persistent search tree produced is $\xi_u$.  This takes $O(\log n)$ amortized time and $O(1)$ amortized space.  After the construction of $\xi_u$ is completed, we can enumerate the values $\xi_u(a)$ for $a \in [1,i-1]$ in linear time for the construction of $A_u$.
	
	The number of nodes in $T$ is $O(t_0|J|)$, and therefore, the storage required by $T$ and its auxiliary structures is $O(t_0|J|)$.   The construction of the trie $T$ and its auxiliary structures takes $O(\log n)$ amortized time for each input number in $I|_J$ in each instance $I_s$ in the training phase.  Therefore, the total construction time is $\tilde{O}(N|J|)$.
	
	\vspace{4pt}
	
	\emph{Analysis of Guarantees}:  In addition to $\beta_i$ and $\tilde{\rho}_i$ for $i \in [1,M]$, we also define the following quantities as in Lemma~\ref{lem:trie-runningtime}:
	
	\smallskip
	
	\begin{minipage}[t]{\textwidth}
	\begin{itemize}
		\item $\rho_i=\Pr_{I \sim \D} \bigl[f(I|_J)=\beta_i\bigr]$ for $i \in [1,M]$, and
		\item $\rho_0 = \Pr_{I \sim \D} \bigl[f(I|_J)\notin \left\{\beta_1,\ldots,\beta_M\right\}\bigr]$.
	\end{itemize}
	\end{minipage}

	\smallskip
	
	We have already discussed in Cases~1 and~2 above that the construction time of $T$ and its auxiliary structures is $\tilde{O}(N|J|)$ and the space complexity of the resulting data structure is $O(t_0|J|)$.  It remains to analyze the expected time to query $T$ when given $I|_J$ for an instance $I$ and the success probability.

	First, we claim that if the outcome $f(I|_J) \not\in \bigl\{\beta_1,\ldots,\beta_M\bigr\}$, the expected query time is $O(|J|\log n)$.   In this case, $f(I|_J)$ is not represented by any leaf of $T$ and the search in $T$ must fail.  As discussed in Cases~1 and~2 above, it takes $O(\log N)$ worst-case time to access $A_u$ for each node of $T$ visited.  Therefore, the total search time is $O(|J|\log N)$ before failure.  Then, computing $f(I|_J)$ from scratch takes another $O(|J|\log n)$ as discussed in Cases~1 and~2 above.   Since $N = \left\lceil t_0 \ln t_0 \ln n \right\rceil$ and $t_0 = O(n|J|)$ by Lemma~\ref{lem:trie-runningtime}, we have $\log N = O(\log n)$, thereby establishing our claim.
	
	Second, we claim that if the outcome $f(I|_J)$ is equal to $\beta_i$ for some $i \in [1,M]$, the expected query time is $O(|J|+\log (1/\tilde{\rho}_i))$.  As discussed in Cases~1 and~2, for every internal node $u$ of the trie $T$, searching $A_u$ to find a child $v$ of $u$ takes
	$O(\log (\mathrm{weight}(u)/\mathrm{weight}(v)))$ time.  Similarly, accessing a child $w$ of $v$ takes $O(\log (\mathrm{weight}(v)/\mathrm{weight}(w)))$ time.  Note that $\mathrm{weight}(v)$ is cancelled when we sum up these two terms.
	Such pairwise cancellations happen when we sum up the search times along the path from the root of $T$ to the leaf representing $\beta_j$.  Therefore, the total search time is $O(|J| + \log (1/\mathrm{weight}(\beta_i))) = O(|J| + \log (1/\tilde{\rho}_i))$.  This proves our second claim.
	
	Combining the above two claims, the expected query time for computing $f(I|_J)$ in the operation phase is
    \[
        \sum_{i=1}^{M} \rho_i \cdot O\bigl(|J|+\log(1/\tilde{\rho}_i)\bigr) + \rho_0 \cdot O(|J|\log n)=
            O\left( |J|+\sum_{i=1}^{M} \rho_i \log(1/\tilde{\rho}_i) + \rho_0|J|\log n \right),
    \]
    which is $O\bigl(H_J(f) + |J| \bigr)$ with probability at least $1-t_0^{-2}$ according to Lemma~\ref{lem:trie-runningtime}.
\end{proof}

We point out that Fredman~\cite{FREDMAN76} obtained a special case of Theorem~\ref{thm:algorithm-limit-output} when $f = \pi$, all outcomes $\beta_1, \beta_2, \ldots$ are given, and every outcome is equally likely.

\cancel{
The following corollary follows from Lemma~\ref{lem:sort-size} and Theorem~\ref{thm:algorithm-limit-output} immediately.

\begin{corollary}\label{corol:tries-for-b-pi}
For each $G_k$, let $H^{(b)}_k$ and $H^{(\pi)}_k$ denote the entropies of $b(I|_{G_k})$ and $\pi(I|_{G_k})$, respectively.
\begin{enumerate}
\item Using $N=\tilde{O}(c_0n|G_k|)$ instances in the training phase, with probability at least $1-n^{-2}$,
  we can build a data structure such that it outputs $b(I|_{G_k})$ in $O(|G_k|+H^{(b)}_k)$ expected time for $I\sim \D$.
\item Using $N=\tilde{O}(\max(|G_k|^2,n))$ instances in the training phase, with probability at least $1-n^{-2}$,
  we can build a data structure such that it outputs $\pi(I|_{G_k})$ in $O(|G_k|+H^{(\pi)}_k)$ expected time for $I\sim \D$.
\end{enumerate}
\end{corollary}
}

\subsection{Operation phase}

In the training phase, we compute the following information:
\begin{itemize}
	
	\item Apply Corollary~\ref{cor:sort-group} to learn the hidden partition of $[1,n]$ into groups $G_1,G_2,\cdots$.
	
	\item Use $\lceil n^2\ln n \rceil$  input instances to construct the $V$-list.
	
	\item For every group $G_k$, apply Theorem~\ref{thm:algorithm-limit-output} to construct the data structures for the purpose of retrieving $b(I|_{G_k})$ and $\pi(I|_{G_k})$ for every instance $I$.

\end{itemize}
%After learning the hidden partition and the $V$-list and building auxiliary structures for computing all the $b(I|_{G_k})$'s and $\pi(I|_{G_k})$'s, our self-improving algorithm proceeds to its operation phase, which is as follows.
Afterwards, for each input instance $I$ that we encounter in the operation phase, we perform the following steps to sort $I$.
\renewcommand{\labelenumi}{\theenumi.}
\begin{enumerate}

	\item Repeat the following steps for each group $G_k$.  Let $I|_{G_k} = (x_{i_1},\cdots,x_{i_m})$.
	\begin{enumerate}
		\item  Use the data structures in Theorem~\ref{thm:algorithm-limit-output} for $G_k$ to retrieve $\pi(x_{i_1},\cdots,x_{i_m})$ and $b(x_{i_1},\cdots,x_{i_m})$.
		
		\item Use $\pi(x_{i_1},\cdots,x_{i_m})$ to sort $(x_{i_1},\cdots,x_{i_m})$ in $O(m)$ time.  Let $x_{s_1} < \cdots < x_{s_m}$ denote the sorting output.

        \item For $j \in \brac{1,m}$, determine the interval $[v_{r_j}, v_{{r_j}+1})$ that contains $x_{s_j}$ in $O(1)$ time using $b(x_{i_1},\cdots,x_{i_m})$.
		
		\item For $r \in [0,n]$, initialize $Z_r:= \emptyset$.  By a left-to-right scan, break $(x_{s_1}, \cdots, x_{s_m})$ at the boundaries of the intervals $[v_r,v_{r+1})$'s  into contiguous subsequences.  For each contiguous subsequence $\sigma$ obtained, let $[v_r,v_{r+1})$ be the interval that contains $\sigma$ and insert $\sigma$ as a new element into $Z_r$.
	\end{enumerate}

	\item For each $r \in [0,n]$, merge the subsequences in $Z_r$ into one sorted list. \label{step:sort}
	
	\item Concatenate the sorted lists produced in step~\ref{step:sort} in the left-to-right order.
        Return the result of the concatenation as the sorted order of $I$.
\end{enumerate}

We will need the following result proved by Ailon~et~al.

\begin{lemma}[{\cite[Claim~2.3]{SICOMP11selfimp}}]
	\label{lem:ailon}
	Let $D$ be a distribution on a universe $U$, and let $F_1: U \rightarrow X$ and $F_2: U \rightarrow Y$ be two random variables.  Suppose that the function $f$ defined by $f : (I,F_1(I)) \mapsto F_2(I)$ can be computed by a comparison-based algorithm with $C$ expected comparisons (where the expectation is over $D$).  Then $H(F_2) = C + O(H(F_1))$, where all the entropies are with respect to $D$.
\end{lemma}
	
We analyze the limiting complexity of the operation phase.

\begin{theorem}
	\label{thm:sort}
	Under the group product distribution model, there is a self-improving sorter with a limiting complexity of $O(n + H_\mathrm{S})$, where $H_\mathrm{S}$ denotes the entropy of the distribution of the sorting output.   The storage needed by the operation phase is $O(n^3)$.  The training phase processes $\tilde{O}(n^2)$ input instances in $\tilde{O}(n^3)$ time.  The success probability is $1- O(1/n)$.
\end{theorem}
\begin{proof}
Correctness is obvious.  The number of input instances required by the training phase, the running time of the training phase, and the storage required by the operation phase follow from Corollary~\ref{cor:sort-group}, Lemma~\ref{lem:sort-size}, and Theorem~\ref{thm:algorithm-limit-output}.  It remains to analyze the limiting complexity and the success probability.

%For every group $G_k$, let $H_k$ denote $H_{G_k}(b_{|G_k|})$ and let $H'_k$ denote $H_{G_k}(\pi_{|G_k|})$.

Step~1(a) runs in $O\bigl(n + \sum_k H_{G_k}(b)+ \sum_k H_{G_k}(\pi) \bigr)$ expected time by Theorem~\ref{thm:algorithm-limit-output}.  Since we can merge the sorted order of $I$ with the $V$-list in $O(n)$ time to obtain $b(I|_{G_k})$'s for all $G_k$, Lemma~\ref{lem:ailon} implies that the joint entropy of $b(I|_{G_1})$, $b(I|_{G_2})$, $\ldots$ is $O(n + H_\text{S})$.  By the group production distribution model, the hidden parameters $u_k$'s follow a product distribution.  Therefore, the joint entropy of $b(I|_{G_1})$, $b(I|_{G_2})$, $\ldots$ is equal to $\sum_k H_{G_k}(b)$.  Hence,
\[
\sum_k H_{G_k}(b) = O(n + H_\text{S}).
\]
As there are $O(|G_k|^2)$ distinct outcomes in $\left\{\pi(I|_{G_k}) : I \sim \D \right\}$ by Lemma~\ref{lem:sort-size}, $H_{G_k}(\pi) = O(\ln |G_k|) = O(|G_k|)$, and so
\[
\sum_k H_{G_k}(\pi) = O(n).
\]
The expected running time of step~1(a) is thus $O(n + H_\text{S})$.

Clearly, Steps~1(b)--(d) and Step~3 take $O(n)$ time.

Consider Step~2.  Let $\sigma_{k,r}$ denote $I|_{G_k} \cap [v_r,v_{r+1})$.
Step~2 runs in $O(\sum_r \sum_k |\sigma_{k,r}|\log|Z_r|) = O\bigl(\sum_r\sum_k  |Z_r|\cdot |\sigma_{k,r}|\bigr)$ time.
Observe that $|Z_r| \leq y_{k,r} + 1$,
  where $y_{k,r}$ denotes the number of groups other than $G_k$ that have elements in $[v_r,v_{r+1})$.
Therefore,
 \begin{align*}
    \E{|Z_r|\cdot |{\sigma}_{k,r}|}
   & \leq \E{(y_{k,r}+1)|\sigma_{k,r}|} = \E{|\sigma_{k,r}|}+ \E{y_{k,r}\cdot |\sigma_{k,r}|} \\
   & = \E{|\sigma_{k,r}|} + \E{y_{k,r}} \E{|\sigma_{k,r}|} = \left(1+\E{y_{k,r}}\right) \E{|\sigma_{k,r}|} \\
   & = O\left(\E{|\sigma_{k,r}|}\right).
 \end{align*}
 The last step above applies the relation $\E{y_{k,r}} = O(1)$, which is implied by Lemma~\ref{lemma:construct-V-sorting}.  Therefore, the expected running time of Step~2 is $O\bigl(\sum_r\sum_k \E{|Z_r|\cdot |\sigma_{k,r}|}\bigr)=O\bigl(\sum_k\sum_r \E{|\sigma_{k,r}|} \bigr)= O(n)$.

We conclude that the total expected running time is $O(n + H_\text{S})$.

By Corollary~\ref{cor:sort-group}, the hidden partition is learned with probability at least $1-O(1/n)$.  The $V$-list satisfies the property in Lemma~\ref{lemma:construct-V-sorting} with probability at least $1-O(1/n)$.  By Theorem~\ref{thm:algorithm-limit-output} and the union bound, the constructions of the data structures that support the retrieval of $b(I|_{G_k})$ and $\pi(I|_{G_k})$ for all groups $G_k$'s succeed with probability at least $1-O(1/n)$.  Hence, the overall success probability is $1-O(1/n)$.
\end{proof}

\section{Self-improving Delaunay triangulator}

We use $\mathcal{G} = (G_1, G_2, \ldots)$ to denote the hidden partition of $[1,n]$ into groups.  Each group $G_k$ is governed by a hidden parameter $u_k \in \R^2$.  The parameters $u_1,u_2,\ldots$ follow a product distribution.  An input instance $I$ consists of $n$ points, $(p_1,\cdots,p_n)$, where the $x$- and $y$-coordinates of $p_i$ are denoted by $p_{i,x}$ and $p_{i,y}$, respectively.  For every $G_k$ and every $i \in G_k$, the group product distribution stipulates that $p_{i,x} = h_{i,k}^x(u_k)$ and $p_{i,y} = h_{i,k}^y(u_k)$ for some hidden functions $h_{i,k}^x$ and $h_{i,k}^y$.

We say that two input points $p_i$ and $p_j$ are in the same group $G_k$ if and only if $i$ and $j$ belong to $G_k$.  Given two input coordinates $\xi_1$ and $\xi_2$, we also say that they are in the same group $G_k$ if and only if $\xi_1 = p_{i,x}$ or $p_{i,y}$ and $\xi_2 = p_{j,x}$ or $p_{j,y}$ for some possibly non-distinct $i,j \in G_k$.

We require the functions $h_{i,k}^x$'s and $h_{i,k}^y$'s to satisfy the properties given in Definition~\ref{df:property} below, where the notion of a $d$-generic distribution is given in Definition~\ref{df:generic}.
\begin{definition}\label{df:property}
	\hspace*{.2in}
	\renewcommand{\theenumi}{\alph{enumi}}
	\renewcommand{\labelenumi}{(\theenumi)}
	\begin{enumerate}
		
		\item For every group $G_k \in \mathcal{G}$ and every $i \in G_k$, the functions $h_{i,k}^x$ and $h_{i,k}^y$ are bivariate polynomials in the $x$- and $y$-coordinates of $u_k$, respectively.  The degrees of these bivariate polynomials are at most some known constant $d_0$.
		
		\item For every group $G_k \in \mathcal{G}$, the distribution of $u_k$ is $d_0d_1$-generic, where $d_1 = 2(d_0^2/2 + d_0)^{16}$.
		
	\end{enumerate}
\end{definition}

\begin{definition}
	\label{df:generic}
	For every integer $d \geq 0$, a distribution $D$ of points in $\R^m$ is \emph{$d$-generic} if for every non-zero $m$-variate polynomial $f : \R^m \rightarrow \R$ of degree at most $d$, $\Pr_{z \sim D} [f(z)=0] =0$.
\end{definition}

It turns out that it is impossible to learn the hidden partition exactly in the case of DT.  We illustrate this difficulty for different distributions of four points.

Let $u_1$ and $u_2$ be uniformly distributed over $[0,1]^2$.  Define four points $p_i$ for $i \in [1,4]$ as follows: $p_1=(u_{1,x},0)$, $p_2=(u_{1,y},0)$, $p_3=(u_{2,x},0)$, and $p_4=(u_{2,y},0)$.  The hidden partition of $[1,4]$ is thus $\bigl(\{1,2\},\{3,4\} \bigr)$.  However, the distribution of $(p_1,p_2,p_3,p_4)$ remains the same even if we redefine the four points as $p_1=(u_{1,x},0)$, $p_2=(u_{2,x},0)$, $p_3=(u_{1,y},0)$, and $p_4=(u_{2,y},0)$, and in this case, the hidden partition of $[1,4]$ is $\bigl(\{1,3\},\{2,4\}\bigr)$.  Similarly, let $u_1$, $u_2$, $u_3$, and $u_4$ be uniformly distributed over $[0,1]^2$, if we define $p_1=(u_{1,x},0)$, $p_2=(u_{2,x},0)$, $p_3=(u_{3,y},0)$, and $p_4=(u_{4,y},0)$, the distribution of $(p_1,p_2,p_3,p_4)$ is identical to the two previous distributions, and in this case, the hidden partition is $\bigl(\{1\},\{2\},\{3\},\{4\}\bigr)$.

%Since $u_1$, $u_2$, $u_3$, and $u_4$ are hidden variables and we can only observe $(p_1,p_2,p_3,p_4)$, we cannot distinguish among the above three settings.
Therefore, it is impossible  to learn the hidden partition exactly.  Nevertheless, the partition $\bigl(\{1\},\{2\},\{3\},\{4\}\bigr)$ is an approximation of all three possibilities.  We define what we mean by an \emph{approximate partition}.
%By our group product distribution, every input coordinate is determined by a non-constant function of the corresponding hidden parameter.  However, for the completeness of the theory of learning an approximation partition, we allow the other case as well which motivates the following definition.

%We are ready to give the definition of an approximate partition of $[1,n]$.

\begin{definition}
	\label{df:partition}
	Let $\mathcal{G} = (G_1, G_2,\ldots)$ be the hidden partition of $[1,n]$ induced by $\D$.  We call another partition $\mathcal{G}' = (G'_0, G'_1,G'_2,\ldots)$ of $[1,n]$ an \emph{approximate partition} if the following conditions are satisfied:
	\renewcommand{\theenumi}{\alph{enumi}}
	\renewcommand{\labelenumi}{(\theenumi)}
	\begin{enumerate}
		
		\item $G'_0$ is the set of indices $i \in [1,n]$ such that both $h_{i,k}^x$ and $h_{i,k}^y$ are constant functions, where $G_k$ is the group in $\mathcal{G}$ that contains $i$.
		
		\item For every $j \geq 1$, there exists $k$ such that $G'_j \subseteq G_k$.
		
		\item Every $G_k \in \mathcal{G}$ contains at most two distinct groups among $G'_1, G'_2, \cdots$.
		
	\end{enumerate}
As a shorthand, we use $G'_{+}$ to denote $[1,n] \setminus G'_0 = G'_1 \cup G'_2 \cup \cdots$.
\end{definition}

Despite the introduction of an approximate partition, when we refer to the functions $h_{i,k}^x$ and $h_{i,k}^y$, the second subscript $k$ refers to the index of the group $G_k$ in the true hidden partition $\mathcal{G}$.

%We will show that it suffices to learn an approximate partition of $[1,n]$ for obtaining a self-improving DT algorithm.

\subsection{Learning an approximate partition for DT}\label{subsect:approx-partition}

Recall that the functions $h_{i,k}^x$'s and $h_{i,k}^y$'s are bivariate polynomials of degrees at most some constant $d_0$, and that the distribution of every $u_k$ is $d_0d_1$-generic, where $d_1 = 2(d_0^2/2 + d_0)^{16}$.  We define a notion of constant input coordinate and prove that every non-constant input coordinate has zero probability of being equal to any particular value.

\begin{definition}
	\label{df:coord}
	An input coordinate $\xi$ is a \emph{constant input coordinate} if the function  that determines $\xi$ is a constant function.  Otherwise, we call $\xi$ a \emph{non-constant input coordinate}.  An input point $p_i$ is a \emph{constant input point} if both of its coordinates are constant input coordinates; otherwise, $p_i$ is a \emph{non-constant input point}.
\end{definition}

\begin{lemma}
	\label{lemma:two-cases-for-one-pier}
	Every input coordinate $\xi$ satisfies exactly one of the following conditions.
	\renewcommand{\theenumi}{\roman{enumi}}
	\renewcommand{\labelenumi}{(\theenumi)}
	\begin{enumerate}
		\item $\xi$ is a constant input coordinate.
		\item For all $c \in \R$, $\Pr \bigl[\xi=c \bigr]=0$.
	\end{enumerate}
\end{lemma}
\begin{proof}
	Without loss of generality, assume that $\xi=h_{i,k}^x(u_k)$.   If $h_{i,k}^x$ is a constant function, then~(i) is true.
    Suppose in the following that $h_{i,k}^x$ is not a constant function and we show that~(ii) holds.
    Take an arbitrary value $c \in \R$.
    Because $h_{i,k}^x$ is not a constant function, $h_{i,k}^x-c$ is a non-zero bivariate polynomial of degree at most $d_0$.
    By Definition~\ref{df:property}(b), the distribution of $u_k$ is $d_0d_1$-generic, which by Definition~\ref{df:generic} implies that $\Pr \bigl[[h_{i,k}^x-c](u_k)=0\bigr]=0$.
    This implies that $\Pr \bigl[h_{i,k}^x(u_k)=c\bigr]=0$, i.e., $\Pr \bigl[\xi=c\bigr]=0$.
\end{proof}

Every pair of non-constant input coordinates must be either \emph{coupled} or \emph{uncoupled} as explained in the following result.  We defer its proof to Appendix~\ref{sec:two-cases-for-two-piers}.

\begin{lemma}\label{lemma:two-cases-for-two-piers}
	For every pair of non-constant input coordinates $\xi_1$ and $\xi_2$, exactly one of the following conditions is satisfied.
	\renewcommand{\theenumi}{\roman{enumi}}
	\renewcommand{\labelenumi}{(\theenumi)}
	\begin{enumerate}
		\item There exists a non-zero bivariate polynomial $f$ of degree at most $d_1$ such that $f(\xi_1,\xi_2)\equiv 0$.
		We say that $\xi_1$ and $\xi_2$ are coupled in this case.
		\item The distribution of $(\xi_1,\xi_2)$ is $d_1$-generic.  We say that $\xi_1$ and $\xi_2$ are uncoupled in this case.
	\end{enumerate}
In particular, if $\xi_1$ and $\xi_2$ are in different groups in $\mathcal{G}$, then $\xi_1$ and $\xi_2$ must be uncoupled.
\end{lemma}

We state two properties of every triple of input coordinates in the following two lemmas.  We defer their proofs to Appendices~\ref{sec:case-1-for-3-piers} and~\ref{sec:case-2-for-3-piers}.

\begin{lemma}\label{lemma:case-1-for-3-piers}
	For every triple of input coordinates $\xi_1$, $\xi_2$, and $\xi_3$ that are in the same group in $\mathcal{G}$, there exists a non-zero trivariate polynomial $f$ with degree at most $d_1$ such that $f(\xi_1,\xi_2,\xi_3)\equiv 0$.
\end{lemma}

\begin{lemma}\label{lemma:case-2-for-3-piers}
	For every triple of non-constant input coordinates $\xi_1$, $\xi_2$, and $\xi_3$ that are pairwise uncoupled and are not in the same group in $\mathcal{G}$, the distribution of $(\xi_1,\xi_2,\xi_3)$ is $d_1$-generic.
\end{lemma}

The following lemma is the main tool for learning an approximate partition.
 %It allows us to test whether two non-constant input coordinates are coupled or uncoupled.
We defer its proof to Appendix~\ref{sec:learn-partition-DT-by-linearity}.

\begin{lemma}\label{lemma:learn-partition-DT-by-linearity}
	Let $m$ and $d$ be two positive integer constants.  Let $\xi_1,\ldots,\xi_m$ be $m$ input coordinates.  Suppose that exactly one of the following conditions is satisfied.
	\renewcommand{\theenumi}{\roman{enumi}}
	\renewcommand{\labelenumi}{(\theenumi)}
	\begin{enumerate}
		\item There exists a non-zero $m$-variate polynomial $f$ of degree at most $d$ such that $f(\xi_1,\ldots,\xi_m)\equiv 0$.
		\item The distribution of $(\xi_1,\ldots,\xi_m)$ is $d$-generic.
	\end{enumerate}
	Then, using $\kappa={m+d\choose m}$ input instances and $O(\kappa^3)$ time, we can determine almost surely whether condition (i) or (ii) is satisfied.
\end{lemma}

By combining Lemmas~\ref{lemma:two-cases-for-one-pier}--\ref{lemma:learn-partition-DT-by-linearity}, we can perform three kinds of tests on input coordinates as stated in the following result.

\begin{lemma}\label{cor:learn-DT}
	\hspace*{.2in}
	\renewcommand{\theenumi}{\roman{enumi}}
	\renewcommand{\labelenumi}{(\theenumi)}
	\begin{enumerate}
		\item For every input coordinate $\xi$, using ${1+1\choose 1}=2$ input instances, we can determine almost surely whether $\xi$ is a constant input coordinate in $O(1)$ time.
		
		\item For every pair of non-constant input coordinates $\xi_1$ and $\xi_2$, using ${2+d_1\choose 2}$ input instances, we can determine almost surely whether $\xi_1$ and $\xi_2$ are coupled or uncoupled in $O(d_1^6) $ time.
		
		\item For every triple of non-constant input coordinates $\xi_1$, $\xi_2$, and $\xi_3$ that are pairwise uncoupled, using ${3+d_1\choose 3}$ input instances, we can determine almost surely whether $\xi_1$, $\xi_2$, and $\xi_3$ are in the same group in $\mathcal{G}$ in $O(d_1^9)$ time.
	\end{enumerate}
\end{lemma}
\begin{proof}
	Consider (i).
By Lemma~\ref{lemma:two-cases-for-one-pier}, either $\xi$ is a constant input coordinate or $\Pr\bigl[\xi=c\bigr]=0$ for all $c \in \R$.
The former case is equivalent to the existence of a non-zero univariate polynomial $f$ of degree~1 such that $f(\xi) \equiv 0$, and the latter case is equivalent to the distribution of $\xi$ being 1-generic.
As a result, by Lemma~\ref{lemma:learn-partition-DT-by-linearity}, we can distinguish between these two cases almost surely using
	${1+1\choose 1} = 2$ input instances in $O(1)$ time.

	The correctness of (ii) follows directly from Lemmas~\ref{lemma:two-cases-for-two-piers} and~\ref{lemma:learn-partition-DT-by-linearity}.

	Consider (iii).  Either one of the following two possibilities hold.
	\begin{itemize}
		\item If $\xi_1$, $\xi_2$, and $\xi_3$ are in the same group in $\mathcal{G}$, then by Lemma~\ref{lemma:case-1-for-3-piers}, there exists a non-zero trivariate polynomial of degree at most $d_1$ such that $f(\xi_1,\xi_2,\xi_3)\equiv 0$.
		\item If $\xi_1$, $\xi_2$, and $\xi_3$ are not in the same group in $\mathcal{G}$, then by Lemma~\ref{lemma:case-2-for-3-piers}, the distribution of $(\xi_1,\xi_2,\xi_3)$ is $d_1$-generic.
	\end{itemize}
	By Lemma~\ref{lemma:learn-partition-DT-by-linearity}, these two possibilities can be distinguished almost surely using ${3+d_1\choose 3}$ instances in $O(d_1^9)$ time.
\end{proof}

We are ready to show that we can construct an approximate partition of $[1,n]$ fairly efficiently.

\begin{lemma}\label{lem:test-one-two-piers}
	Using $O(n^3)$ input instances, we can find an approximate partition of $\,[1,n]$ almost surely in $O(n^3\alpha(n))$ time.
\end{lemma}

\begin{proof}
	We construct an approximate partition $\mathcal{G}'$ as stated in Definition~\ref{df:partition}.  First, we apply Lemma~\ref{cor:learn-DT}(i) to decide for each input coordinate whether it is a constant input coordinate or not.  This step requires $O(n)$ input instances and $O(n)$ time.
	
	Afterwards, we can pick out the constant input points.  The indices of these constant input points form $G_0'$ in $\mathcal{G}'$.  For $j \geq 1$, we initialize each $G'_j$ to contain a distinct index in $[1,n] \setminus G_0'$.   This gives the initial $\mathcal{G}' = (G_0',G_1',G_2',\ldots)$.  Properties~(a) and (b) in Definition~\ref{df:partition} are satisfied, but property~(c) may be violated.   We show below how to merge groups in $\mathcal{G'}$  step by step so that properties~(a) and~(b) are preserved and property~(c) is satisfied in the end.  Note that $G_0'$ will not change throughout the successive merges.
	
	There are two phases in the merging process.
	
	Take two non-constant input coordinates $\xi_1$ and $\xi_2$ that are in distinct groups in the current $\mathcal{G}'$.   We apply Lemma~\ref{cor:learn-DT}(ii) to test whether $\xi_1$ and $\xi_2$ are coupled.  If the test reveals that $\xi_1$ and $\xi_2$ are coupled, by Lemma~\ref{lemma:two-cases-for-two-piers}, $\xi_1$ and $\xi_2$ must be in the same group in $\mathcal{G}$, and therefore, we update $\mathcal{G}'$ by merging the groups in $\mathcal{G}'$ that contain $\xi_1$ and $\xi_2$.  We repeat this step until no two groups of $\mathcal{G}'$ can be merged anymore.  This completes phase one.  We require $O(1)$ input instances and $O(1)$ time to test a pair of input coordinates.  The rest can be done using a disjoint union-find data structure~\cite{MIT}.  So a total of $O(n^2)$ input instances and $O(n^2\alpha(n))$ time are needed in phase one.
	
	Take a triple of non-constant input coordinates $\xi_1$, $\xi_2$, and $\xi_3$ that are in three distinct groups in the current $\mathcal{G}'$.  Thanks to phase one, $\xi_1$, $\xi_2$, and $\xi_3$ are pairwise uncoupled as they are in distinct groups in $\mathcal{G}'$.  Therefore, Lemma~\ref{cor:learn-DT}(iii) is applicable and we use it to decide whether $\xi_1$, $\xi_2$, and $\xi_3$ are in the same group in $\mathcal{G}$.  If so, we update $\mathcal{G}'$ by merging the groups in $\mathcal{G}'$ that contain $\xi_1$, $\xi_2$, and $\xi_3$.  We repeat this step until no three groups of $\mathcal{G}'$ can be merged anymore.  This completes phase two.  We require $O(1)$ instances and $O(1)$ time to test a triple of non-constant input coordinates.  So a total of $O(n^3)$ instances and $O(n^3\alpha(n))$ time are needed in phase two.
	
	Clearly, properties~(a) and~(b) in Definition~\ref{df:partition} are preserved throughout.  Suppose for the sake of contradiction that there are at least three groups $G'_i$, $G'_j$ and $G'_l$ in $\mathcal{G}'$, where $i$, $j$ and $l$ are positive, that are contained in the same group $G_k \in \mathcal{G}$ after phase two.  Since the input points in $G'_i$, $G'_j$ and $G'_l$ are not constant input points, we can find three non-constant input coordinates $\xi_1$, $\xi_2$ and $\xi_3$ in $G'_i$, $G'_j$ and $G'_l$, respectively.  But then, the application of Lemma~\ref{cor:learn-DT}(iii) to $\xi_1$, $\xi_2$ and $\xi_3$ in phase two must have told us to merge $G'_i$, $G'_j$ and $G'_l$, a contradiction.  This shows that property~(c) in Definition~\ref{df:partition} is satisfied after phase two.
\end{proof}

\subsection{A canonical set $V$ and its Delaunay triangulation}
\label{sec:V}
	
Recall that $p_1,\ldots,p_n$ denote the $n$ input points in an instance.  We first show that the probability of any non-constant input point $p_i$ lying in certain algebraic curves in $\R^2$ is zero.

\begin{lemma}\label{lemma:circle-zero-probability}
For every $i \in [1,n]$, if $p_i$ is a non-constant input point, then for every circle $\mathcal{O}$ in $\R^2$, it holds that  $\Pr [p_i\in \mathcal{O}]=0$.
\end{lemma}
\begin{proof}
Without loss of generality, assume that $i$ is in group $G_k$ in $\mathcal{G}$.
Then, $(p_{i,x},p_{i,y})=(h^{x}_{i,k}(u_k),h^{y}_{i,k}(u_k))$.
Assume $u_k=(x,y)$ and recall that $h_{i,k}^x$ and $h_{i,k}^y$ are polynomials of $x,y$ of degree at most $d_0$.
The relation $p_i\in \mathcal{O}$ is described by an equation $(p_{i,x}-a)^2+(p_{i,y}-b)^2 = c^2$ for some $a,b,c\in \R$.
After substitution, this equation becomes $g(x,y)=0$ for some polynomial $g$ of degree at most $2d_0$.
We show that $g$ is a non-zero polynomial in the following.

Because $p_i$ is a non-constant input point, there exists $(i,j)$ such that $i\geq 0,j\geq 0,i+j>0$ and that
   $x^iy^j$ is a monomial of $h^{x}_{i,k}(x,y)$ or $h^{y}_{i,k}(x,y)$ whereas $x^{i'}y^{j'}$
      is neither a monomial of $h^{x}_{i,k}(x_y)$ nor a monomial of $h^{y}_{i,k}(x,y)$ for every $(i',j')$ such that $i'\geq i,j'\geq j',i'+j'>i+j$.
   This means $x^{2i}y^{2j}$ must be a monomial of $g$ and hence $g$ is not the zero polynomial.

Further since the distribution of $u_k=(x,y)$ is $d_0d_1$-generic (and so $2d_0$-generic), Definition~\ref{df:generic} implies that $\Pr[g(x,y)=0]=0$.
In other words, $\Pr [p_i\in \mathcal{O}]=0$.
\end{proof}

Let $\mathcal{G}' = (G_0',G_1',G_2',\ldots)$ be an approximate partition of $\, [1,n]$ produced by applying Lemma~\ref{lem:test-one-two-piers}.  Recall from Definition~\ref{df:partition} that 	$G'_{+} = [1,n] \setminus G'_0 = G'_1\cup G'_2 \cup \ldots$.  That is, for every input instance $I$, the subsequence $I|_{G'_{+}}$ consists of all non-constant input points.

Let $\lambda=\lceil n^2\ln n\rceil$.  Take $\lambda$ instances $I_1,\ldots,I_\lambda$.  Let $S$ be set of the $\lambda n$ points in these instances.  Build a $(1/n)$-net $V'$ of $S$ with respect to disks, that is, for every disk $C$, $|C \cap S| \geq |S|/n \Rightarrow C \cap V' \not= \emptyset$.  It is known that $|V'| = O(n)$~\cite{clarkson07}.  Add to $V'$ three special points that form a huge triangle $\tau$ such that any input point lies inside $\tau$.  Let $V$ be the union of $V'$ and these three special points.  The canonical Delaunay triangulation $\del(V)$ satisfies Lemma~\ref{lemma:construct-V-set} below.  Its proof  is analogous to that of Lemma~\ref{lemma:construct-V-sorting}.

\begin{lemma}
	\label{lemma:construct-V-set}
	%With probability at least $1-n^{-189}$ over our construction of $V$,
	It holds with probability at least $1-1/n^{189}$ that for every triangle $t$ in $\del(V)$,
        \[\mathrm{E}_{I\sim \D}\bigl[\bigl|\, I|_{G'_{+}} \cap C_t \, \bigr|\bigr]=O(1),\]
      where $C_t$ denotes the circumscribing disk of $t$.
	%for every $r$ in $\{0,\ldots,n\}$, we have $E_I(|I \cap [V_r,V_{r+1})|)=O(1)$.
	%(Here, $[V_0,V_1)$ denotes $(V_0,V_1)=(-\infty,V_1)$.)
\end{lemma}
\begin{proof}
	%The proof is similar to that of Lemma~\ref{lemma:construct-V-sorting}.  We sketch it below.
    Let $V = \{q_{-2}, q_{-1}, q_0, q_1, \cdots, q_{\lambda n} \}$, where $\{q_1,\ldots,q_{\lambda n}\}=S$ and $q_0$, $q_{-1}$, and $q_{-2}$ are the special points that define the huge triangle $\tau$.

    Fix a triple $(i,j,k)$ for some distinct indices $i,j,k \in [-2,\lambda n]$.  Let $m_{ijk}$ be the number of instances among $I_1,\ldots,I_\lambda$ that do not contain any of $q_i$, $q_j$, and $q_k$.  Clearly, $m_{ijk} \in [\lambda-3,\lambda]$.
    Denote these instances by $I_{s_1},\cdots,I_{s_{m_{ijk}}}$.  Let $D_{ijk}$ denote the open circumscribing disk of $q_i$, $q_j$ and $q_k$.   As $D_{ijk}$ is an open set, none of $q_i$, $q_j$ and $q_k$ belongs to $D_{ijk}$.

	For $a \in [1,m_{ijk}]$, define the random variable $Y^{(i,j,k)}_a = \bigl|I_{s_a} \cap D_{ijk} \bigr|$.  Define $Y^{(i,j,k)}=Y^{(i,j,k)}_1+\ldots+Y^{(i,j,k)}_{m_{ijk}}$.  Using the same argument in the proof of Lemma~\ref{lemma:construct-V-sorting}, we can show that it holds with probability at least $1-n^{-199}$ that $\E{Y^{(i,j,k)}} \leq 11 \lambda \text{ or } Y^{(i,j,k)}>\lambda$.
    The expectation $\E{Y^{(i,j,k)}}$ is taken over the distribution of $I_{s_1},\ldots, I_{s_{m_{ijk}}}$ induced by the group product distribution $\D$.

	There are at most $(\lambda n+3)^3 < n^{10}$ triples of distinct indices $i, j, k \in [-2,\lambda n]$.  By the union bound, it holds with probability at least $1-n^{-189}$ that for every distinct triple $i,j,k \in [-2,\lambda n]$, $\E{Y^{(i,j,k)}} \leq 11 \lambda$ or $Y^{(i,j,k)}>\lambda$.  Assume for the rest of this proof that this event happens.
	
    Take an arbitrary triangle $t$ in $\del(V)$.  Let $q_a$, $q_b$, and $q_c$ denote the vertices of $t$.  As $t$ is a Delaunay triangle,  $V\cap D_{abc} = \emptyset$.  As $V$ is a $(1/n)$-net of $S$ for disks, $V \cap D_{abc} = \emptyset$ implies that $|S\cap D_{abc}|<\frac{1}{n}|S|=\lambda$.
    Therefore, fewer than $\lambda$ points from $I_{s_1},\cdots, I_{s_{m_{ijk}}}$ fall into $D_{abc}$, i.e., $Y^{(a,b,c)}<\lambda$.  Since we assume the event that for every distinct triple $i,j,k \in [-2,\lambda n]$, $\E{Y^{(i,j,k)}} \leq 11 \lambda$ or $Y^{(i,j,k)}>\lambda$, it must be the case that $\E{Y^{(a,b,c)}} \leq 11\lambda$.

    By Lemma~\ref{lemma:circle-zero-probability}, the probability of a non-constant input point lying on the circumcircle of $t$ is zero.  It follows that $\mathrm{E}_{I\sim \D} \bigl[\bigl|\,I|_{G'_{+}} \cap C_t\,\bigr|\bigr]=\E{Y^{(a,b,c)}}/m_{abc} \leq \E{Y^{(a,b,c)}}/(\lambda - 3) =O(1)$.
    %Hence, $\mathrm{E}_{I\sim \D} |I' \cap c_{ijk}|=O(1)$.
    %As a consequence of Lemma~\ref{lemma:circle-zero-probability}, $\mathrm{E}_{I\sim \D} |I' \cap C_t|=\mathrm{E}_{I\sim \D} |I' \cap c_{ijk}|$.
    %Together, $\mathrm{E}_{I\sim \D} |I' \cap C_t|=O(1)$.
\end{proof}

\subsection{Functions for point location and retrieval of split trees}

We will need to define two functions that are analogous to $b$ and $\pi$ for sorting.  Let $|\del(V)|$ denote the number of triangles in $\del(V)$.  Let the triangles in $\del(V)$ be denoted by $t_1,\ldots,t_{|\del(V)|}$.   For every input instance $I = (p_1,\ldots,p_n)$, the first function $B$ returns the triangles in $\del(V)$ that contains the $p_i$'s.

\begin{quote}
$B : \bigcup_{m=1}^n \R^{2m} \rightarrow  \bigcup_{m=1}^n \brac{1,|\del(V)|}^m$ such that for every $m$-tuple of points $(q_1,\cdots,q_m)$, $B(q_1,\cdots,q_{m}) = (r_1,\cdots,r_m)$, where $t_{r_i}$ is the triangle in $\del(V)$ that contains $q_i$ for $i\in [1,m]$.
\end{quote}

The second function to be defined is to output a particular kind of \emph{fair split tree}~\cite{CK95} for every input instance.  The fair split tree will help us to produce the Delaunay triangulation of the input instance efficiently.  The fair split tree variant that we need is defined as follows.

Let $P$ be any set of points in $\R^2$.
Let $R$ be the smallest axis-aligned bounding rectangle of $P$.  Let $\hat{R}$ be the smallest square that contains $R$ and has the same center as $R$.  We initialize the split tree to be a single node $v$ with $R(v) = R$ and $\hat{R}(v) = \hat{R}$.   The point set at $v$ is $P$.  In general, we grow the split tree by splitting a node $w$ whose point set consists of more than one point.  Inductively, $R(w)$ is the smallest axis-aligned bounding rectangle of the point set at $w$, and $\hat{R}(w)$ is an axis-aligned outer rectangle that contains $R(w)$.  To split $w$, take the bisecting line of $R(w)$ that is perpendicular to a longest side of $R(w)$.  This line splits the point set at $w$ into two non-empty subsets, and it also splits $\hat{R}(w)$ into the outer rectangles of the children of $w$.  The expansion of the split tree bottoms out at nodes that represent only one point in $P$.   This gives a split tree of $P$ that is a full binary tree with $|P|$ leaves.  As we bisect $R(w)$ for each node $w$ to generate its two children, we call this particular kind of fair split tree the \emph{halving split tree} of $P$.

For every rectangle $r$, let $\ell_{\min} (r)$ and $\ell_{\max} (r)$ be the minimum and maximum side lengths of $r$, respectively.  The  following property is satisfied by a fair split tree~\cite[equation~(1) in Lemma~4.1]{CK95}:
\begin{equation}
	\text{For each node $u$ of the split tree}, \,\, \ell_{\min}(\hat{R}(u)) \geq \frac{1}{3}\ell_{\max}(R(\text{parent}(u))).  \label{eq:1}
\end{equation}
There are fair split trees that satisfy \eqref{eq:1} but are not halving split trees.

We use $\ST(P)$ to denote a fair split tree of $P$, not necessarily the halving split tree.  Using \eqref{eq:1}, it can be shown that a fair split tree can be used to produce a well-separated pair decomposition of $O(|P|)$ size~\cite{CK95}, which can then be used to produce a Delaunay triangulation in $O(|P|)$ expected time as we explain in Section~\ref{sec:split-DT} later.
%Enforcing \eqref{eq:1} is the reason why a smallest bounding square, instead of a rectangle, is chosen as $\hat{R}(u) = R(u)$ at the root $u$ of the split tree.%a nearest-neighbor graph in $O(n)$ time~\cite{CK95}.  There is a randomized algorithm that constructs a Delaunay triangulation from the nearest-neighbor graph in $O(n)$ expected time~\cite{DTbySort}.

We are ready to define the second function $\Pi$.  Let $\text{HST}_m$ denote the set of halving split trees of all possible point sets in $\R^2$ with size $m$.
\footnote{Two halving split trees $T_1,T_2$ are regarded equivalent
   if (1) the underlying trees are the same and (2) the set of points to be split on each node $v$ are the same between $T_1$ and $T_2$
      and (3) the direction of the splitting line at each internal node $v$ in is the same between $T_1$ and $T_2$.
   Therefore, when the splitting lines are at different positions but the other structures of the halving split trees $T_1,T_2$ are the same,
     $T_1,T_2$ are regarded equivalent. Only one representative in each equivalent class is included in $\text{HST}_m$.
     Unless otherwise stated, we assume that the halving split trees do not include the information about the positions of the splitting lines.}

\begin{quote}
$\Pi : \bigcup_{m=1}^n \R^{2m} \rightarrow  \bigcup_{m=1}^n [0,n]^{m} \times [0,n]^m \times \text{HST}_m$ such that for every $m$-tuple of points $(q_1,\cdots,q_m)$,  the output of $\Pi(q_1,\cdots,q_m)$ is an ordered triple that consists of the following data:
\begin{itemize}
	\item  $\pi(q_{1,x}, \cdots, q_{m,x})$,
	\item $\pi(q_{1,y}, \cdots, q_{m,y})$, and
	\item the halving split tree of $\{q_1,\cdots,q_m\}$.
\end{itemize}
\end{quote}
The encodings $\pi(q_{1,x}, \cdots, q_{m,x})$ and $\pi(q_{1,y}, \cdots, q_{m,y})$ induce sorted orders $q_1,\ldots,q_m$ by $x$- and $y$-coordinates, respectively.  We call the corresponding permutations of $[1,m]$ the \emph{$x$-order} and \emph{$y$-order} of $q_1,\ldots,q_m$, respectively.

For every subset $S \subseteq \mathbb{R}^{2m}$, define $B(S) = \{B(x) : x \in S\}$ and $\Pi(S) = \{\Pi(x) : x \in S\}$.  We first bound the cardinalities of $B(S)$ and $\Pi(S)$ when $S = \{ I|_{G'_j}  :  I \sim \D\}$.

\begin{lemma}\label{lem:DT-size}
Assume $J$ is subset of some group $G_k$ in $\mathcal{G}$.
Let $S = \bigl\{I|_J : I \sim \D \bigr\}$.  Let $m = |J|$.  Then,
	\[|B(S)| = O(n^2m^2)\text{ and }|\Pi(S)| = O(m^8).\]
\end{lemma}
\begin{proof}
  Without loss of generality, assume that $J=[1,m]$.

  We first deal with $B(S)$.  Let $\mathcal{L}$ be the set of support lines of all edges in $\del(V)$.  Let the equations of these lines be
  $\alpha_s x+\beta_s y+\gamma_s=0$ for $s \in  [1,|{\mathcal{L}}|]$.  Consider the algebraic curves in $\R^2$ defined by the following equations:
  \[
  \forall\, i \in [1,m], \, \forall\, s \in [1,|{\mathcal{L}}|], \quad \alpha_s h^x_{i,k}(u_k) + \beta_s h^y_{i,k}(u_k) + \gamma_s=0.
  \]

  Let $\mathcal{A}_B$ be the arrangement of these $O(nm)$ curves.  Since these curves have degrees no more than $d_0$, the complexity of $\mathcal{A}_B$ is $O(n^2m^2)$.  Within any cell of $\mathcal{A}_B$, the sign of each function $ \alpha_s h^x_{i,k}(u_k) + \beta_s h^y_{i,k}(u_k) + \gamma_s$ is invariant.  These signs determine $B(p_1,\ldots,p_m)$.  Hence, $|B(S)| = O(n^2m^2)$.

  Next, we deal with $\Pi(S)$.  Consider the following equations for some possibly non-distinct indices $i,j,r,s\in [1,m]$:
\begin{center}
\begin{minipage}{0.3\textwidth}
	   \begin{align*}
   		h^x_{i,k} &=h^x_{j,k}, \\
        h^y_{i,k}+h^y_{j,k} &= 2 h^y_{r,k}, \\
        h^x_{i,k} - h^x_{j,k} &= h^y_{r,k} - h^y_{s,k}.
   \end{align*}
\end{minipage}
\begin{minipage}{0.3\textwidth}
	\begin{align*}
	h^y_{i,k} &= h^y_{j,k}, \\
	h^x_{i,k}+h^x_{j,k} &= 2h^x_{r,k}, \\
    \end{align*}
\end{minipage}
\end{center}
Each of these equations defines an algebraic curve in $\R^2$ of degree at most $d_0$, so the arrangement $\mathcal{A}_\Pi$ of these $O(m^4)$ curves has complexity $O(m^8)$.  We will prove that $\Pi(S)=\Pi(p_1,\ldots,p_m)$ is invariant within a cell of ${\mathcal{A}}_\Pi$.  Therefore, the number of distinct outputs of $\Pi(S)$ is no more than the number of  cells of $\mathcal{A}_\Pi$.  It follows that $|\Pi(S)| = O(m^8)$.

Let $C$ be an arbitrary cell of $\mathcal{A}_\Pi$.

First, the $x$-order of $p_1,\ldots, p_m$ is invariant within $C$.  Otherwise, when $u_k$ moves within $C$ and triggers a change in the $x$-order of $p_1,\ldots,p_m$, there must be a setting of $u_k$ such that  $h^x_{i,k}(u_k)=h^x_{i',k}(u_k)$ for some $i,i' \in [1,m]$.  But then the interior of $C$ intersects one of the curves, a contradiction.  Similarly, the $y$-order of $p_1,\ldots, p_m$ is also invariant within $C$.

Let $p_{i_1}$ and $p_{i_2}$ denote the points with the smallest and largest $x$-coordinates within $C$, respectively, among $p_1,\ldots,p_m$.  Let  $p_{i_3}$ and $p_{i_4}$ denote the points with the smallest and largest $y$-coordinates within $C$, respectively, among $p_1,\ldots,p_m$.

Let $v$ denote the root of the halving split tree of $\{p_1,\ldots,p_m\}$.  Since the equation $(h^x_{i_2,k} - h^x_{i_1,k}) = (h^y_{i_4,k} - h^y_{i_3,k})$ is included in defining $\mathcal{A}_\Pi$,  the sign of $(h^x_{i_2,k} - h^x_{i_1,k}) - (h^y_{i_4,k} - h^y_{i_3,k})$ is invariant within $C$.   This sign determines whether the node $v$ should be split vertically or horizontally.

Suppose that $v$ is split vertically.  The curves $h^x_{i_1,k}+h^x_{i_2,k} =2h^x_{i,k}$ for $i\in [1,m]\setminus \{i_1,i_2\}$ are included in defining $\mathcal{A}_\Pi$.  Therefore, for every $i\in [1,m]\setminus \{i_1,i_2\}$, the sign of $(h^x_{i_1,k}+h^x_{i_2,k})/2 -h^x_{i,k}$ is invariant within $C$.  For every $i\in [1,m]\setminus \{i_1,i_2\}$, this sign determines whether $p_i$ lies to the left or right of the splitting line of $R(v)$.   As a result, the indices of the points at the two children of $v$ are completely determined by the choice of $C$.  Similarly, if $v$ is split horizontally, we can conclude that for every $i\in [1,m]\setminus \{i_3,i_4\}$, either $p_i$ stays above  the splitting line of $R(v)$ for all $u_k \in C$, or $p_i$ stays below it for all $u_k \in C$, which implies that the indices of the points at the two children of $v$ are also completely determined by the choice of $C$.

Repeating the argument above as the halving split tree grows leads to the conclusion that $\Pi(S)$ is invariant within $C$.  This establishes the $O(m^8)$ bound on $|\Pi(S)|$.
\end{proof}

Next, we generalize Theorem~\ref{thm:algorithm-limit-output} to work for $B$ and $\Pi$.   It is the main tool for achieving our self-improving DT algorithm.

\begin{theorem}\label{thm:algorithm-limit-output-2}
	Let $J$ be a subset of $\, [1,n]$.    Let $f$ be $B$ or $\Pi$.
    Let $N = \left\lceil t_0 \ln t_0 \ln n\right\rceil$, where $t_0=\mathrm{poly}(n)$ is
      a given value that bounds the size of $\,\bigl|\{ f(I|_J) : I \sim \D \}\bigr|$.
    We can build a data structure such that, given $I|_J$ for an input instance $I$, the data structure returns $f(I|_J)$ and it has the following performance guarantees.
	\begin{itemize}
		\item It can be constructed in $\tilde{O}(N|J|^2)$ time using $N$ independent $I_1|_J, \ldots, I_N|_J$.
		\item It uses $O(t_0|J|^2)$ space.
		\item It holds with probability at least $1 - t_0^{-2}$ that the expected query time is $O\bigl( H_J(f) + |J|\bigr)$, where the expectation is taken over $I \sim \D$.
	\end{itemize}
\end{theorem}
\begin{proof}
Without loss of generality, assume that $J = [1,m]$ and hence $I|_J=(p_1,\cdots,p_m)$.  Let $I_1, \ldots, I_N$ denote the independent input instances that we sample for the construction of the data structure.  Define the following quantities:
\begin{itemize}
	\item $\beta_1, \dots, \beta_M$ denote the distinct outcomes in $\bigl\{ f(I_s|_J : s \in [1,N]\bigr\}$, and
	\item $\tilde{\rho}_i = \frac{1}{N} \cdot$ number of occurrences of $\beta_i$ among $f(I_1|J), \ldots, f(I_N|_J)$ for $i \in [1,M]$.
	%\item $\rho_i=\Pr_{I \sim \D} \bigl[f(I|_J)=\beta_i\bigr]$ for $i \in [1,M]$, and
	%\item $\rho_0 = \Pr_{I \sim \D} \bigl[f(I|_J)\notin \left\{\beta_1,\ldots,\beta_M\right\}\bigr]$.
\end{itemize}

\vspace{4pt}

\emph{Case~1. $f=B$}: We build a trie $T$ like in the case of $f = b$ in the proof of Theorem~\ref{thm:algorithm-limit-output}.  If we are at $p_{i-1}$ and a node $u$ of $T$, we want to obtain an auxiliary structure at $u$ after the training phase that gives us the child $w$ of $u$ such that $uw$ has the label $r \in [1,|\del(V)|]$ for which triangle $t_r$ in $\del(V)$ contains $p_i$.

In the training phase, we keep one worst-case optimal planar point location structure for $\del(V)$~(e.g.~\cite{paper:Kirkpatrick1981}).  Whenever we are at $p_{i-1}$ and a node $u$ of the trie $T$ being grown, we query the point location structure with $p_i$ to identify the triangle $t_r \in \del(V)$ that contains $p_i$.  If $u$ has no edge to a child with $r$ as the label, we create a new child $w$ of $u$ and assign $r$ to be the label of $uw$.  Therefore, we spend $O(\log n)$ at each node in processing an input instance in the training phase.

At the end of the training phase, we take the subset of triangles represented by the children of $u$, triangulate the exterior of these triangles, and then construct a distribution-sensitive planar point location structure for the resulting triangulation~\cite{iacono04}.  The weights of nodes in $T$ are defined in the same way as in the proof of Theorem~\ref{thm:algorithm-limit-output}.  For an edge $uw$ with label $r$, the weight of $t_r$ in this distribution-sensitive point location structure is equal to the weight of $w$ in $T$.  The weights of other triangles in the triangulation are set to zero.  This construction time is $O(\log n)$ times the number of children of $u$.

In summary, the total time spent in the training phase is $O(N|J|\log n)$.  The space needed by the trie $T$ and the auxiliary structures is $O(t_0|J|)$.

In the operation phase, when we use an input point $p_i$ to search the distribution-sensitive point location structure at $u$, if $p_i$ lies in a triangle represented by a child $w$ of $u$, the search takes $O(\log (\text{weight}(u)/\text{weight}(w)))$; otherwise, an unsuccessful search takes $O(\log N)$ time~\cite{iacono04}.  An unsuccessful search triggers the computation of $b(I|_J)$ from scratch in $O(|J|\log n)$ time.
	
\vspace{4pt}

\emph{Case 2. $f = \Pi$}:  We first explain the structure of the trie $T$ to be built.  Then, we discuss how to build it in the training phase.  Afterwards, we explain how to search in $T$ in the operation phase.

{\sc Structure}: The top $m$ levels of the trie is a smaller trie $T_1$ for determining $\pi(p_{1,x},\cdots,p_{m,x})$.  So $T_1$ is a copy of the trie for $\pi$ in the proof of Theorem~\ref{thm:algorithm-limit-output}.  We expand each leaf $u$ of $T_1$ into a trie $T_{2,u}$ for determining $\pi(p_{1,y},\cdots,p_{m,y})$.  Let $T_2$ be the composition of $T_1$ and all $T_{2,u}$'s.  Given an instance $I\sim \D$, we know the $x$-order and $y$-order of $(p_1,\ldots,p_m)$ at a leaf of $T_2$.
	
Each leaf $v$ of $T_2$ is expanded to a trie $T_{3,v}$ with the following properties.  Define the point set $P_v$ at $v$ to be $\{p_1,\ldots,p_m\}$.  A leaf $w$ of $T_{3,v}$ corresponds to the halving split tree $S$ of $\{p_1,\ldots,p_m\}$, depending on the relative positions of these points in an input instance.    The sequence of nodes on the path $\varrho$ in $T_{3,v}$ from $v$ to the leaf $w$ corresponds to the preorder traversal of $S$.   For every node $u$ in $\varrho$, we use $\eta_u$ to denote the node in $S$ that corresponds to $u$.  The node $u$ is associated with a subset $P_u$ of $\{p_1,\ldots,p_m\}$ so that $P_u$ is the point set at $\eta_u$ in $S$.

Take an arbitrary node $u$ in $T_{3,v}$.  For $i \in [1,|P_u|-1]$, $u$ may have a child corresponding to a vertical cut between the $i$-th and $(i+1)$-th points in the $x$-order of $P_u$.  Similarly, for $i \in [1,|P_u|-1]$, $u$ may have a child corresponding to a horizontal cut between the $i$-th and $(i+1)$-th points in the $y$-order of $P_u$.   So $u$ has at most $2|P_u|-2$ children.  The existence of a particular child $w$ of $u$ depends on whether some input instance in the training phase prompts the creation of $w$.

The final trie is the composition of $T_2$ and all $T_{3,v}$'s.   We use $T_3$ to denote this final trie.

A leaf of $T_3$ that corresponds to an outcome $\beta_i$ for some $i \in [1,M]$ is given a weight of $\tilde{\rho}_i$.  Every internal node $u$ of $T_3$ is given a weight $\text{weight}(u)$ equal to the sum of the weights of leaves that descend from $u$.

Each internal node $u$ in $T_1$ keeps a map $\xi_u : [0,i-1] \rightarrow [0,i-1]$, assuming that the path from the root of $T_1$ to $u$ has $i-1$ edges on it, such that $\xi_u(0) = 0$ and $\xi_u(1),\ldots,\xi_u(i-1)$ is the $x$-order of $p_1,\ldots,p_{i-1}$.  The node $u$ also has a nearly optimal binary search tree $A_u$ that stores $\xi_u(0),\ldots,\xi_u(i-1)$ in this order such that each element of $A_u$ stores either a null pointer or a pointer to a distinct child $w$ of $u$.  In the former case, that element of $A_u$ has weight zero, and in the latter case, that element of $A_u$ has weight $\text{weight}(w)$.
%For each child of $u$, there is exactly one element of $A_u$ that points to it; however, the pointers at some elements of $A_u$ may be null.

Similarly, each internal node $w$ in $T_{2,u}$ for some leaf $u$ of $T_1$ keeps a map $\xi_w : [0,i-1] \rightarrow [0,i-1]$, assuming that the path from $u$ to $w$ in $T_{2,u}$ has $i-1$ edges in it, such that $\xi_u(0) = 0$ and $\xi_u(1),\ldots,\xi_u(i-1)$ is the $y$-order of $p_1,\ldots,p_{i-1}$.  The node $u$ also has a nearly optimal binary search tree $A_u$ that stores $\xi_u(0),\ldots,\xi_u(i-1)$ in this order as described in the previous paragraph.

Hence, both the $x$-order and $y$-order of $p_1,\ldots,p_{m}$ are fixed at the root of $T_{3,v}$ for each leaf $v$ of $T_2$.  As mentioned earlier, each internal node $u$ in $T_{3,v}$ correspond to a node $\eta_u$ in some halving split tree $S$.  And $u$ keeps a subset $P_u$ of $\{p_1,\ldots,p_m\}$ such that $P_u$ is the point set at $\eta_u$ in $S$.  We store two maps $\xi_{u,x}$ and $\xi_{u,y}$ such that $\xi_{u,x}(1), \xi_{u,x}(2),\ldots$ and $\xi_{u,y}(1), \xi_{u,y}(2),\ldots$ are the $x$-order and $y$-order of the point set $P_u$ at $u$, respectively.  We store the $x$-order $\xi_{u,x}$ in a nearly optimal binary search tree $A_{u,x}$ and the $y$-order $\xi_{u,y}$ in another nearly optimal binary search tree $A_{u,y}$ as mentioned in the previous two paragraphs.  Again, each element in $A_{u,x}$ and $A_{u,y}$ stores either a null pointer or a pointer to a distinct child of $u$.  For each child of $u$, there is exactly one element in $A_{u,x}$ and $A_{u,y}$ that points to it; however, the pointers at some elements of $A_{u,x}$ and $A_{u,y}$ may be null.

Because of the auxiliary structures at each node of the $T_{3,v}$'s, we use $O(|J|)$ space per node.  There are $O(t_0|J|)$ nodes in $T_3$.  The overall space complexity is thus $O(t_0|J|^2)$.

{\sc Training phase}:  The construction of $T_1$ and $T_2$ and the auxiliary structures at their internal nodes are done in the same way as in the proof of Theorem~\ref{thm:algorithm-limit-output}.   This takes $O(\log n)$ amortized time and $O(1)$ amortized space for each input point in each instance $I_s$, $s \in [1,N]$.

Suppose that $I_s|_J$ for some $s \in [1,N]$ leads us to a leaf $v$ of $T_2$.  We navigate and possibly grow the sub-trie $T_{3,v}$ as follows.  First, we construct the halving split tree $S$ for $I_s|_J$ in $O(m\log m)$ time.  Then, we obtain a preorder traversal of $S$ in $O(m)$ time.   We navigate down $T_{3,v}$ and scan the preorder traversal of $S$ in a synchronized manner.   Recall that both the $x$-order and $y$-order of $P_v$ are already determined.

In general, suppose that we reach a node $u$ of $T_{3,v}$.  We keep a balanced binary search tree $L_{u,x}$ that stores the $x$-order $\xi_{u,x}(1), \xi_{u,x}(2), \cdots$.  If we encounter a vertical cut in the training phase that splits $P_u$ between $\xi_{u,x}(i)$ and $\xi_{u,x}(i+1)$, then the node for $\xi_u(i)$ in $L_{u,x}$ stores the pointer to a child $w$ of $u$. Otherwise, the node for $\xi_{u,x}(i)$ in $L_{u,x}$ stores a null pointer.  The tree $L_{u,y}$ is similarly organized for the $y$-order of $P_u$.  Recall that $\eta_u$ denotes the node in the split tree $S$ that corresponds to $u$.  Suppose that the cut at $\eta_u$ is vertical with $x$-coordinate $X$.  We search $L_{u,x}$ to locate $X$ between $\xi_{u,x}(i)$ and $\xi_{u,x}(i+1)$ for some $i$.  If the node for $\xi_{u,x}(i)$ in $L_{u,x}$ stores a pointer to a child $w$ of $u$, we go to $w$ and the next node in the preorder traversal in $S$.   Otherwise, we create a new child $w$ of $v$ and store a pointer to $w$ at  the node for $\xi_{u,x}(i)$ in $L_{u,x}$.  Let $\gamma$ be node that follows $\eta_u$ in the preorder traversal of $S$.   We set $\eta_w = \gamma$ and $P_w$ to be the point set at $\gamma$.  We also initialize $L_{w,x}$ and $L_{w,y}$ to store the $x$-order and $y$-order of $P_w$, respectively.  Initially, every node in $L_{w,x}$ and $L_{w,y}$ stores a null pointer because $w$ has no child yet.

%Each node in the split tree $S$ is already associated with the $x$-order and $y$-order of the point set associated with it during the construction of $S$.  To make the initialization of $L_{w,x}$ and $L_{w,y}$ efficient, we store the $x$-order and $y$-order of point sets at nodes in $S$ using persistent search trees.  That is, the $x$-order and $y$-order of $\{p_1,\ldots,p_m\}$ are stored using a persistent search tree at the root of $S$.  During the construction of a split tree, the $x$-order and $y$-order of the point set at a new node $\eta$ are obtained by extracting the relevant entries from $x$-order and $y$-order of its parent.   The very same operations are performed on the persistent search trees of the parent of $\eta$ to obtain the persistent search trees at $\eta$.  This keeps the construction time of $S$ at $O(m \cdot \mathit{polylog}(m))$.   Now, initializing $L_{w,x}$ and $L_{w,y}$ boils down to making new roots for the new versions fo the two persistent search trees at the node $\eta'$ in $S$.  This takes $O(1)$ amortized time and amortized space.

We can similarly handle the case that the cut at $\eta$ is horizontal.

At the end of the training phase, we construct nearly optimal binary search trees $A_{u,x}$ and $A_{u,y}$ to store the contents of $L_{u,x}$ and $L_{u,y}$, respectively.  For each node in $L_{u,x}$ that points to a child $w$ of $u$ in $L_{u,x}$, its weight in $A_{u,x}$ is set to be $\text{weight}(w)$.   For each node in $L_{u,x}$ that has a null pointer, its weight in $A_{u,x}$ is set to be zero.

In summary, for each input point in each $I_s|_J$, $s \in [1,N]$, we spend $O(\log n)$ amortized time in the construction of $T_2$ and $O(|J|\log |J|) = O(|J|\log n)$ time in the construction of some $T_{3,v}$.  The overall time complexity of the training phase is thus $O(N|J|^2\log n)$.

{\sc Operation phase}:  The search of $A_u$ at an internal node $u$ of $T_1$ or $T_2$ is conducted as in the proof of Theorem~\ref{thm:algorithm-limit-output}.  So accessing a child $w$ of $u$ takes $O(\log (\text{weight}(u)/\text{weight}(w)))$ time, and an unsuccessful search in $A_{u}$ takes $O(\log N)$ time.

Consider the search at an internal node $u$ of $T_{3,v}$ for some leaf $v$ of $T_2$ in the operation phase.  Let $\kappa = |P_u|$.  By comparing $p_{\xi_{u,x}(\kappa),x} - p_{\xi_{u,x}(1),x}$ and $p_{\xi_{u,y}(\kappa),y} - p_{\xi_{u,y}(1),y}$ in $O(1)$ time, we can determine whether $P_u$ should be split vertically or horizontally for the given input instance, as well as the $x$- or $y$-coordinate of the splitting line.  If the splitting line of $P_u$ is vertical and its $x$-coordinate is $X$, we search in $A_{u,x}$ to divide the $x$-order of $P_u$ at $X$.   We abort the search in $A_{u,x}$ if we come to a node such that every leaf descending from it contains a null pointer.  By doing so, if the search leads to a child $w$ of $u$, the search time is $O(\log (\text{weight}(u)/\text{weight}(w)))$; otherwise, the search is unsuccessful and it takes $O(\log N)$ time.  Symmetrically, if the splitting line of $P_u$ is horizontal and its $y$-coordinate is $Y$, we search in $A_{u,y}$ to divide the $y$-order of $P_u$ at $Y$.

\vspace{4pt}

\emph{Analysis of Guarantees}: Since we preserve the time to descend from a node of the trie to its appropriate child as in the proof of Theorem~\ref{thm:algorithm-limit-output}, we can invoke the same analysis to conclude that the overall expected search time is $O(|J| + H_J(f))$ with probability at least $1 - t_0^{-2}$.   The construction time at each node of the trie $T$ is $\tilde{O}(|J|)$ in the case of $f = \Pi$ which dominates the other case of $f = B$.  Furthermore, the space complexity at each node of $T$ is $O(|J|)$ in the case of $f = \Pi$ which also dominates the other case of $f = B$.   Therefore, the overall construction time is $\tilde{O}(N|J|^2)$ and the storage required by $T$ and the auxiliary structures is $O(t_0|J|^2)$.
\end{proof}

\cancel{
The following corollary follows from Lemma~\ref{lem:DT-size} and Theorem~\ref{thm:algorithm-limit-output-2} immediately.

\begin{corollary}\label{corol:tries-for-B-Pi}
For each $l\geq 1$, denote by $H^{(B)}_l,H^{(\Pi)}_l$ the entropies of $B(I|_{G^*_l}),\Pi(I|_{G^*_l})$, respectively.
\begin{enumerate}
\item Using $N=\tilde{O}(n^2|G^*_l|^2)$ instances in the training phase, with probability at least $1-n^{-2}$,
  we can build a data structure such that it outputs $B(I|_{G^*_l})$ in $O(|G^*_l|+H^{(B)}_l)$ expected time for $I\sim \D$.
\item Using $N=\tilde{O}(\max(|G^*_l|^8,n))$ instances in the training phase, with probability at least $1-n^{-2}$,
  we can build a data structure such that it outputs $\Pi(I|_{G^*_l})$ in $O(|G^*_l|+H^{(\Pi)}_l)$ expected time for $I\sim \D$.
\end{enumerate}
\end{corollary}
}

\subsection{From split tree to Delaunay triangulation}
\label{sec:split-DT}

We show in this subsection how to construct $\del(P)$ from $\ST(P)$ in $O(|P|)$ expected time for any set of points $P$ in $\R^2$.
Recall that $\ST(P)$ denotes a split tree of $P$ that satisfies \eqref{eq:1}, but it is not necessarily the halving split tree of $P$.

Given two point sets $P$ and $Q$ such that $Q \subseteq P$, we can construct $\ST(Q)$ from $\ST(P)$ as follows.  Make a copy $T$ of $\ST(P)$.  Remove all leaves of $T$ that represent points in $P \setminus Q$.  Repeatedly remove nodes in $T$ with only one child until there is none.   For each node $u$ of $\ST(P)$ that is inherited by $\ST(Q)$, the same cut that splits $u$ in $\ST(P)$ is also used in splitting $u$ in $\ST(Q)$.   So $\ST(Q)$ may not be the halving split tree of $Q$.  For every surviving node $u$ in $\ST(Q)$, $R(u)$ may shrink due to point deletions, but $\hat{R}(u)$ may remain the same or expand.  For example, if the parent of $u$ is deleted but the grandparent of $u$ survives, then $\hat{R}(u)$ in $\ST(Q)$ is equal to  $\hat{R}(\text{parent}(u))$ in $\ST(P)$.   The key is that \eqref{eq:1} is satisfied by $\ST(Q)$.

\begin{lemma}\label{lemma:compute-inherent-tree}
	Suppose that we have built a $\ST(P)$ for a point set $P$.
	\renewcommand{\theenumi}{\roman{enumi}}
	\renewcommand{\labelenumi}{(\theenumi)}
	\begin{enumerate}
		\item For every subset $Q\subseteq P$, $\ST(Q)$ can be computed from $\ST(P)$ in $O(|P|)$ time.
		\item For every $m$ subsets $Q_1, \cdots, Q_m$ of $P$, if each $Q_i$ is ordered as in the preorder traversal of $\ST(P)$, then
		$\ST(Q_1), \cdots, \ST(Q_m)$ can be computed from $\ST(P)$ in $O\bigl(\alpha(|P|) \cdot \bigl(|P| + \sum_{i=1}^m |Q_i|\bigr)\bigr)$ time, where $\alpha(\cdot)$ is the inverse Ackermann function.
	\end{enumerate}
\end{lemma}
\begin{proof}
	Claim (i) follows from our previous discussion.  For (ii), the construction of $\ST(Q_i)$ boils down to $O(|Q_i|)$ nearest common ancestor queries in $\ST(P)$~\cite[Thm~3.8]{DTbySort}, which can be solved in the time stated~\cite{Offline-NCA-Tarjan}.  There are solutions without the factor $\alpha(\cdot)$, but they require table lookup which is incompatible with the comparison-based model here.
\end{proof}

 The following connection between $\ST(P)$ and $\del(P)$ follows from~\cite{CK95,DTbySort}.

\begin{lemma}\label{lemma:fst-DT}
  There is a randomized algorithm that, given $\ST(P)$ of a point set $P$, constructs the Delaunay triangulation $\del(P)$ in $O(|P|)$ expected time.
\end{lemma}

\begin{proof}
Let $\mathit{NN}(Q)$ denote the nearest-neighbor graph of a point set $Q$.  A randomized algorithm is described in~\cite{DTbySort} for constructing $\del(P)$ using nearest neighbor graphs.  This algorithm is recursive in nature.  First, initialize $P_0$ to be $P$.  Then, randomly sample a subset $P_1$ from $P_0$.  It is shown that $P_1$ satisfies some useful properties after trying an expected $O(1)$ random draws from $P_0$.  In particular, $|P_1|$ is less than $|P_0|$ by a constant factor.  Then $\del(P_1)$ is recursively computed.  And finally, $\del(P_1)$ is ``merged'' with $\mathit{NN}(P_0)$ to produce $\del(P_0)$.

In general, we need to randomly sample $P_{i+1}$ from $P_i$, recursively construct $\del(P_{i+1})$ and then merge it with $\mathit{NN}(P_i)$ to form $\del(P_i)$.   Drawing the random sample $P_{i+1} \subset P_i$ takes $O(|P_i|)$ worst-case time.  Assuming that $\mathit{NN}(P_i)$ is available for all $i \geq 0$ during the recursive calls, this randomized algorithm takes $O(|P_0|+|P_1|+\ldots)=O(n)$ expected time~\cite{DTbySort}.

What remains to be proved is that once $\ST(P)$ is given, we can compute $\mathit{NN}(P_0)$, $\mathit{NN}(P_1)$, $\ldots$ in $O(n)$ expected time.  Given any point set $Q$, it has been shown that $\mathit{NN}(Q)$ can be constructed in $O(|Q|)$ worst-case time from $\ST(Q)$~\cite{CK95}.  Therefore, we can generate $\mathit{NN}(P_0)$ from $\ST(P_0)$ in $O(|P_0|)$ time.  For $i \geq 1$, we construct $\ST(P_i)$ from $\ST(P_{i-1})$ in $O(|P_{i-1}|)$ time by Lemma~\ref{lemma:compute-inherent-tree}~(i), and then we compute $\mathit{NN}(P_i)$ from $\ST(P_i)$ in $O(|P_i|)$ time.  The total running time for computing $\mathit{NN}(P_0)$, $\mathit{NN}(P_1)$, $\ldots$ is $O(|P_0|+|P_1|+\ldots)=O(n)$.
\end{proof}

\subsection{Operation phase}

In the training phase, we compute the following information:
\begin{itemize}
	
	\item Apply Lemma~\ref{lem:test-one-two-piers} to learn an approximate partition $\mathcal{G}' = (G'_0,G'_1,\ldots)$ of $[1,n]$.
	
	\item Take an input instance $I$ and compute $\del(I|_{G'_0})$.  Note that changes in $I$ do not affect $\del(I|_{G'_0})$ because  $I|_{G'_0}$ consists of constant input points.
	
	\item Use $\lceil n^2\ln n \rceil$  instances to construct the point set $V$ of size $n$ and then construct $\del(V)$.
	
	\item For every $j \geq 1$, apply Theorem~\ref{thm:algorithm-limit-output-2} to construct the data structures for the purpose of retrieving $B(I|_{G'_j})$ and $\Pi(I|_{G'_j})$ for every instance $I$.
	
\end{itemize}
%After learning the hidden partition and the $V$-list and building auxiliary structures for computing all the $b(I|_{G_k})$'s and $\pi(I|_{G_k})$'s, our self-improving algorithm proceeds to its operation phase, which is as follows.
Afterwards, for each input instance $I$ that we encounter in the operation phase, we perform the following steps to construct $\del(I)$.  Recall that $G'_{+}$ denotes $[1,n] \setminus G'_0 = G'_1 \cup G'_2 \cup \cdots$.
\begin{enumerate}
	\item For each $j\geq 1$, use the data structures in Theorem~\ref{thm:algorithm-limit-output-2} for $G'_j$ to retrieve $B(I|_{G'_j})$ and $\Pi(I|_{G'_j})$.

	\item The results $B(I|_{G'_1}),B(I|_{G'_2}),\ldots$ determine, for each $p_i \in I|_{G'_{+}}$, the triangle $t$ in $\del(V)$ that contains $p_i$.
        For each $p_i\in I|_{G'_{+}}$, do a breadth-first-search in $\del(V)$ from $t$ to compute
        \[
        \Delta_i = \{t \in \del(V) : p_i \in C_t\},
        \]
        where $C_t$ is the circumscribing disk of $t$.
        The breadth-first-search can be applied here as the triangles in $\Delta_i$ are connected, as proved in \cite{SICOMP11selfimp}.
	
	\item For each $j \geq 1$, perform the following steps.
	\begin{enumerate}
		\item Take the halving split tree $\ST(I|_{G'_j})$ from the output $\Pi(I|_{G'_j})$.
		\item Traverse $\ST(I|_{G'_j})$ in preorder to produce an ordered list $Q_j$ of points in $I|_{G'_j}$.
		\item For every triangle $t \in \bigcup_{i \in G'_j} \Delta_i$, compute
		\[
		Q_{j,t}= Q_j \cap C_t.
		\]
		This is done by initializing $Q_{j,t} = \emptyset$ for each $t \in \bigcup_{j \in G'_j} \Delta_i$, followed by a linear scan of $Q_j$ (in order) that appends each $p_i \in Q_j$ scanned to $Q_{j,t}$ for every $t \in \Delta_i$.
	\end{enumerate}

	\item For every $j \geq 1$, apply Lemma~\ref{lemma:compute-inherent-tree}~(ii) to construct $\ST(Q_{j,t})$ for every $t \in \bigcup_{i \in G'_j} \Delta_i$
     from $\ST(I|_{G'_j}) = \ST(Q_j)$.
	
    \item For every $j \geq 1$ and every $t \in \bigcup_{i \in G'_j} \Delta_i$, compute $\del(Q_{j,t})$ from $\ST(Q_{j,t})$ using Lemma~\ref{lemma:fst-DT}.
	%\item For each $t \in \del(V)$, let $Q_t = \bigcup_{l} Q_{l,t}$ and compute $\del(Q_t)$ by merging the non-empty $\del(Q_{l,t})$'s over all $l$.
	
	\item Compute the Voronoi diagram $\vor(V \cup I|_{G'_{+}})$ from the $\del(Q_{j,t})$'s.  This gives $\del(V \cup I|_{G'_{+}})$.
	
	\item Split $\del(V \cup I|_{G'_{+}})$ into $\del(I|_{G'_{+}})$ and $\del(V)$.

    \item Merge $\del(I|_{G'_{+}})$ with $\del(I|_{G'_0})$ to produce $\del(I)$.  Return $\del(I)$.	
\end{enumerate}

%We discuss the implementation and running time of the steps below.

\begin{theorem}
	\label{thm:DT}
    Under the group product distribution model, there is a self-improving Delaunay triangulation algorithm with a limiting complexity of $O(n\alpha(n) + H_\mathrm{DT})$, where $H_\mathrm{DT}$ is the entropy of the distribution of the output Delaunay triangulation.  The storage used in the operation phase is $O(n^{10})$.  The training phase processes $\tilde{O}(n^8)$ instances in $\tilde{O}(n^{10})$ time.  The success probability is $1- O(1/n)$.
\end{theorem}
\begin{proof}
By Lemma~\ref{lem:DT-size} and Theorem~\ref{thm:algorithm-limit-output-2}, the training phase processes $\tilde{O}(n^8)$ instances in $\tilde{O}(n^{10})$ time, and the resulting trie and auxiliary structures use $O(n^{10})$ space.

By Theorem~\ref{thm:algorithm-limit-output-2}, step~1 takes $O\bigl(n + \sum_{j \geq 1} H_{G'_j}(B) + \sum_{j \geq 1} H_{G'_j}(\Pi) \bigr)$ expected time.

Steps~2 and~3 run in $O\bigl(n + \sum_{p_i\in I|_{G'_{+}}} |\Delta_i|\bigr)$ time.
%Calculate the total number of pairs $(p_i,t)$ where $p_i\in I',t\in \del(V)$ such that $p_i\in C_t$ in different ways can lead to:
By definition,
    \[
    {\sum}_{p_i\in I|_{G'_{+}}} |\Delta_i|={\sum}_{t\in\del(V)} \bigl|\, I|_{G'_{+}} \cap C_t \, \bigr|.
    \]
By Lemma~\ref{lemma:construct-V-set}, it holds with probability at least $1 - n^{-189}$ that $\E{\bigl |\, I|_{G'_{+}}\cap C_t \, \bigr|} = O(1)$ for every triangle $t\in \del(V)$.   Therefore, it holds with probability at least $1 - n^{-189}$ that
\begin{equation}
	\mathrm{E}\Bigl[\sum_{p_i\in I|_{G'_{+}}} |\Delta_i|\Bigr]=\sum_{t\in\del(V)} \mathrm{E}\Bigl[\bigl| \, I|_{G'_{+}}\cap C_t \, \bigr| \Bigr]=O(n).
	\label{eq:DT-1}
\end{equation}
Hence, steps~2 and~3 run in $O(n)$ expected time with probability at least $1 - n^{-189}$.

By Lemma~\ref{lemma:compute-inherent-tree}~(ii), the expected running time of step~4 is
\[
\alpha(n)  \cdot O\Bigl(\sum_{j\geq 1} |G'_j|  + \sum_{\substack{j\geq 1 \\ t\in \del(V)}} \E{|Q_{j,t}|} \Bigr)
 = n\alpha(n)  + \alpha(n) \cdot O\Bigl(\sum_{t\in \del(V)} \E{\bigl|\, I|_{G'_{+}}\cap C_t \, \bigr|} \Bigr).
 \]
Therefore, by \eqref{eq:DT-1}, it holds with probability at least $1 - n^{-189}$ that step~4 runs in $O(n\alpha(n))$ expected time.

By Lemma~\ref{lemma:fst-DT}, the expected running time of step~5 is
\[
O\Bigl(\sum_{\substack{j \geq 1 \\ t \in \del(V)}} \E{ |Q_{j,t}| }\Bigr) =
O\Bigl(\sum_{t\in \del(V)} \E{\bigl|\, I|_{G'_{+}}\cap C_t \, \bigr|} \Bigr).
\]
Therefore, by \eqref{eq:DT-1}, step~5 runs in $O(n)$ expected time with probability at least $1 - n^{-189}$.

We will explain step~6 shortly.  For step~7, a randomized algorithm is given in~\cite{DT-split} that splits $\del(V \cup I|_{G'_{+}})$ in $O(n)$ expected time.
Step~8 takes $O(n)$ time because we can merge two Delaunay triangulations in linear time~\cite{DT-merge3,DT-merge1}.
%Recall that $\del(I|_{G'_0})$, which is invariant as points in $I|_{G'_0}$ are constant input points, has already been computed in the training phase.
%(We compute $\del(I|_{G^*_0})$ in the learning phase. This Delaunay triangulation is invariant as the points in $G^*_0$ are invariant).

We analyze step~6 in detail.  We decompose step~6 into two tasks.  For every triangle $t$ in $\del(V)$, let $P_t = I|_{G'_{+}}\cap C_t$.  Note that $P_t = \bigcup_{j\geq 1} Q_{j,t}$.  First, compute $\vor(P_t)$ for each triangle $t$ in $\del(V)$.  Second, construct $\vor(V \cup I_{G'_{+}})$ from the $\vor(P_t)$'s.

The first task of computing $\vor(P_t)$ is accomplished by merging the non-empty $\vor(Q_{j,t})$'s over all $j \geq 1$.  Note that $\vor(Q_{j,t})$ can be obtained from the given $\del(Q_{j,t})$ in linear time.   It is also known how to merge two Voronoi diagrams in linear time~\cite{DT-merge3,DT-merge1}.  So we can merge the non-empty $\vor(Q_{j,t})$'s in $O\bigl(\sum_{j\geq 1} z_t |Q_{j,t}|\bigr)$ time, where $z_t$ is the number of groups $G'_j$'s, $j \geq 1$, with points in $C_t$.  The expected running time is $O\bigl(\E{\sum_{t \in \del(V)} \sum_{j\geq 1} z_t |Q_{j,t}|}\bigr)$.  The analysis of a similar quantity $O\bigl(\E{\sum_r\sum_k  |Z_r||\sigma_{k,r}|}\bigr)$ appeared earlier in the proof of Theorem~\ref{thm:sort}.  The same analysis is applicable here and it gives $\E{\sum_{t \in \del(V)} \sum_{j\geq 1} z_t|Q_{j,t}|} = O\bigl(\sum_{t \in \del(V)} \sum_{j\geq 1} \E{|Q_{j,t}|}\bigr) = O(n)$.

The second task of constructing $\vor(V \cup I_{G'_{+}})$ from the $\vor(P_t)$'s can be accomplished as follows. We adopt the strategy in~\cite{SICOMP11selfimp}.  For each triangle $t$ in $\del(V)$, let $\nu_t$ denote the Voronoi vertex in $\vor(V)$ that is dual to $t$.  For each Voronoi cell $C$ of $\vor(V)$, pick the vertex $\nu_t$ of $C$ such that $|P_t|$ is smallest, breaking ties by selecting $t$ with the smallest id in $\del(V)$, and then triangulate $C$ by connecting $\nu_t$ to other vertices of $C$.  This gives the \emph{geode triangulation} of $\vor(V)$ with respect to $I_{G'_{+}}$.  The resulting triangles are the \emph{geode triangles}.

\begin{center}
\begin{minipage}[c]{0.9\textwidth}
\begin{proposition}[{\cite[Claim~4.5]{SICOMP11selfimp}}]
	\label{lemma:restricted-to-P*_s}
	Let $v_\tau$ be the point in $V$ whose Voronoi cell contains a geode triangle $\tau = \nu_{t_1}\nu_{t_2}\nu_{t_3}$.  Then, $\vor(V\cup I|_{G'_{+}}) \cap \tau = \vor\bigl(\{v_\tau\} \cup \bigcup_{i=1}^3 P_{t_i} \bigr) \cap \tau$.
\end{proposition}
\end{minipage}
\end{center}

By Proposition~\ref{lemma:restricted-to-P*_s}, we can compute $\vor\bigl(\{v_\tau\} \cup \bigcup_{i=1}^3 P_{t_i} \bigr) \cap \tau$ for every geode triangle $\tau$, and then stitch these Voronoi diagram fragments to form $\vor(V \cup I|_{G'_{+}})$.  The expected running time is dominated by the expected construction time of $\vor\bigl(\{v_\tau\} \cup \bigcup_{i=1}^3 P_{t_i} \bigr)$ summed over all $\tau$'s.   Since $\vor(P_{t_1})$,  $\vor(P_{t_2})$, $\vor(P_{t_3})$, and  $v_\tau$ can be merged to form $\vor\bigl(\{v_\tau\} \cup \bigcup_{i=1}^3 P_{t_i} \bigr)$ in linear time~\cite{DT-merge3,DT-merge1}, the expected merging time is $O\bigl(\E{1 + \sum_{i=1}^3 |P_{t_i}|}\bigr) = O(1)$ with probability at least $1 - n^{-189}$ because $\E{|P_{t_i}|}=\E{\bigl|\, I|_{G'_{+}} \cap C_{t_i} \, \bigr|} = O(1)$ is satisfied with this probability bound by Lemma~\ref{lemma:construct-V-set}.

We conclude that step~6 runs in $O(n)$ expected time with probability at least $1 - n^{-189}$.

In summary, it holds with probability at least $1 - n^{-189}$ that the expected running time of the operation phase is
\[
O\Bigl(n\alpha(n) + \sum_{j \geq 1} H_{G'_j}(B) + \sum_{j \geq 1} H_{G'_j}(\Pi) \Bigr).
\]
For every $j \geq 1$, there are $O(|G'_j|^8)$ different outputs of $\Pi(I|_{G'_j})$ by Lemma~\ref{lem:DT-size}.  Therefore,
\[
\sum_{j \geq 1} H_{G'_j}(\Pi) = O\Bigl(\sum_{j \geq 1} \ln |G'_j|\Bigr) = O(n).
\]
Since every group in the hidden partition $\mathcal{G}$ contains at most two distinct groups among $G'_1, G'_2, \cdots$, the following inequality is satisfied:
\begin{equation}
\sum_{j\geq 1} H_{G'_j}(B) \leq 2\sum_{G_k \in \mathcal{G}} H_{G_k}(B).
\label{eq:DT-2}
\end{equation}
%Ailon~et~al.~\cite[Lemma~2.3]{SICOMP11selfimp} proved that, given two functions $F_1$ and $F_2$ over the same universe, if for every input $U$, $F_2(U)$ can be computed from $F_1(U)$ using $O(n)$ comparisons, then the entropy of $F_2$ is asymptotically bounded by the entropy of $F_1$ plus $O(n)$.
Given an input instance $I$ and $\del(I)$, an algorithm is given in~\cite[Section~4.2]{SICOMP11selfimp} for finding the triangles in $\del(V)$ that contain the points in $I$.  The same algorithm also works in our case here.  The expected running time of this algorithm is dominated by %$\E{\sum_{p_i \in I|_{G'_{+}}}|\Delta_i|}$
\[
\sum_{t \in \del(V)} \E{| I \cap C_t|} = \bigl|\, I|_{G'_0} \cap C_t \, \bigr| + \mathrm{E}\Bigl[\sum_{p_i \in I|_{G'_{+}}} |\Delta_i|\Bigr],
\]
which by \eqref{eq:DT-1} is $O(n)$ with probability at least $1 - n^{-189}$.  Therefore, by Lemma~\ref{lem:ailon}, the entropy of the joint distribution of $B(I|_{G_1}), B(I|_{G_2}), \ldots$ over all $G_k \in\mathcal{G}$ is $O(n + H_\text{DT})$ with probability at least $1 - n^{-189}$.  The groups in $\mathcal{G}$ are independent under the group product distribution model, and therefore, $H\bigl(B(I|_{G_1}), B(I|_{G_2}), \ldots\bigr) = \sum_{G_k \in \mathcal{G}} H_{G_k}(B)$.  It follows from \eqref{eq:DT-2} that $\sum_{j\geq 1} H_{G'_j}(B) = O(n + H_\text{DT})$.

Altogether, the limiting complexity is $O(n\alpha(n) + H_\mathrm{DT})$.
\end{proof}

\section{Conclusion}

We introduce the group product distribution to model dependence among input items in designing self-improving algorithms.  This gives the new challenge of learning the hidden grouping of the input items.  It also calls for building new auxiliary structures in the training phase for input items in the same group in order to retrieve precomputed solutions quickly in the operation phase.  We show that these new difficulties can be overcome for sorting and Delaunay triangulation computation.   We achieve the optimal limiting complexity for sorting and a nearly optimal limiting complexity for Delaunay triangulation with a polynomial-time training phase.

%\begin{acks}
%Research of Cheng, Jin and Wong are supported by \grantsponsor{}{Research Grants Council, Hong Kong, China}{} (project %no.~\grantnum[]{}{16200317}).  Research of Chiu is supported by \grantsponsor{}{ERC StG}{}~\grantnum[]{}{757609}.
%\end{acks}

\bibliographystyle{plain}% the recommended bibstyle
\bibliography{J-GSI}

\appendix

\section{Appendix}

Some lemmas declared in subsection~\ref{subsect:approx-partition} are proved in this appendix.

\subsection{Proof of Lemma~\ref{lemma:two-cases-for-two-piers}}
\label{sec:two-cases-for-two-piers}

\cancel{
\begin{lemma}\label{lemma:two-cases-for-two-piers}
	For any two non-constant input coordinates $\xi_1$ and $\xi_2$, exactly one of the following holds.
	\begin{enumerate}
		\item[(i)] $f(\xi_1,\xi_2)\equiv 0$ for some non-zero bivariate polynomial $f$ of degree at most $d_1$.
		\item[(ii)] The distribution of $(\xi_1,\xi_2)$ is $d_1$-generic.
	\end{enumerate}
\end{lemma}
}

%\begin{proof}
	
	There are two cases depending on whether $\xi_1$ and $\xi_2$ are in the same group in $\mathcal{G}$.
	
	Case 1: $\xi_1$ and $\xi_2$ are in different groups in $\mathcal{G}$.   We show that the distribution of $(\xi_1,\xi_2)$ is $d_1$-generic.  Take an arbitrary non-zero bivariate polynomial $f(\xi_1,\xi_2)$ of degree at most $d_1$.  We can express $f(\xi_1,\xi_2)$ as $f_{0}(\xi_1) \cdot \xi_2^{d_1}+f_{1}(\xi_1) \cdot \xi_2^{d_1-1}+\ldots+f_{d_1-1}(\xi_1) \cdot \xi_2+f_{d_1}(\xi_1)$,  where $f_i(\xi_1)$ is a polynomial in $\xi_1$ with degree at most $i$.  Since $f$ is not the zero polynomial, there exists $j \in [0,d_1]$ such that $f_j(\xi_1)$ is not the zero polynomial whereas $f_i(\xi_1)$ is the zero polynomial for all $i \in [0,j-1]$.  Let $R$ be the set of roots of the equation $f_j(\xi_1) = 0$.  By the fundamental theorem of algebra, $|R| \leq j \leq d_1$.
	
	Since $\xi_1$ is a non-constant input coordinate by assumption, Lemma~\ref{lemma:two-cases-for-one-pier} implies that $\Pr [\xi_1 = c] = 0$ for all $c \in \R$.  Therefore,
	\begin{equation}
		\Pr \bigl[\xi_1 \in R] = 0.  \label{eq:diffgroup-0}
	\end{equation}
	Fix an arbitrary $r \in \R \setminus R$.   Then, $f(r,\xi_2)$ becomes a polynomial $\hat{f}(\xi_2)$  of degree at most $d_1$.
	Also, $\hat{f}(\xi_2)$ is not the zero polynomial because the coefficient $f_j(r)$ of the monomial $\xi_2^{d_1-j}$ is non-zero by our choice of $r \in \R \setminus R$.  The equation $\hat{f}(\xi_2) = 0$ has at most $d_1$ roots by the fundamental theorem of algebra.   Since $\xi_1$ and $\xi_2$ are in different groups in $\mathcal{G}$ by the case assumption, $\xi_2$ is independent from $\xi_1$.   Since $\xi_2$ is a non-constant input coordinate, Lemma~\ref{lemma:two-cases-for-one-pier} implies that
	\begin{equation}
		\forall \, r \in \R \setminus R, \quad \Pr \bigl[f(\xi_1,\xi_2) = 0 \, | \, \xi_1 = r\bigr] = 0.  \label{eq:diffgroup}
	\end{equation}
	Let $p : \R \rightarrow \R$ be the probability density function of the distribution of $\xi_1$.  We have
	\begin{eqnarray*}
		\Pr \bigl[f(\xi_1,\xi_2) = 0 \bigr] & = &
		\Pr \bigl[f(\xi_1,\xi_2) = 0 \, \wedge \, \xi_1 \not\in R\bigr] + \Pr \bigl[f(\xi_1,\xi_2) = 0 \, \wedge \, \xi_1 \in R\bigr] \\
		& \stackrel{\eqref{eq:diffgroup-0}}{=} & \Pr \bigl[f(\xi_1,\xi_2) = 0 \, \wedge \, \xi_1 \not\in R\bigr] \\
		& = & \int_{\R \setminus R} p(r) \cdot \Pr \bigl[ f(\xi_1,\xi_2) = 0 \, | \, \xi_1 = r \bigr] \, \mathrm{d}r  \\
		& \stackrel{\eqref{eq:diffgroup}}{=} & 0.
	\end{eqnarray*}

	Hence, the distribution of $(\xi_1,\xi_2)$ is $d_1$-generic by Definition~\ref{df:generic}.
	
	Case 2: $\xi_1$ and $\xi_2$ are in the same group in $\mathcal{G}$.   Conditions in (i) and (ii) in the lemma are logical negations of each other.   Therefore, it suffices to prove that the distribution of $(\xi_1,\xi_2)$ is $d_1$-generic under the assumption that (i) does not hold.
	%That is, we assume that for every non-zero bivariate polynomial $f$ of degree at most $d_1$, $\Pr \bigl[f(\xi_1,\xi_2)=0\bigr]=0$.
	
	Without loss of generality, assume that $\xi_1$ and $\xi_2$ are in the group $G_k \in \mathcal{G}$.  Therefore, $\xi_1=h(u_k)$ and $\xi_2=h'(u_k)$ for some bivariate polynomials of degrees at most $d_0$.  Take an arbitrary non-zero bivariate polynomial $f(\xi_1,\xi_2)$.  Substituting $\xi_1 = h(u_k)$ and $\xi_2 = h'(u_k)$ into $f(\xi_1,\xi_2)$, we obtain a polynomial $\hat{f}(u_k)$ with degree at most $d_0d_1$.  Moreover, $\hat{f}(u_k)$ is not  the zero polynomial---otherwise, $f(\xi_1,\xi_2) \equiv 0$, contradicting our assumption that (i) does not hold.  Then, since the distribution of $u_k$ is $d_0d_1$-generic by Definition~\ref{df:property}(b),  we conclude by Definition~\ref{df:generic} that $\Pr \bigl[ \hat f(u_k)=0 \bigr]=0$ and therefore $\Pr \bigl[ f(\xi_1,\xi_2)=0 \bigr]=0$.
%\end{proof}

\subsection{Proof of Lemma~\ref{lemma:case-1-for-3-piers}}
\label{sec:case-1-for-3-piers}

\cancel{
Given $\xi_1$ and $\xi_2$ as in Lemma~\ref{lemma:two-cases-for-two-piers},
we say $\xi_1,\xi_2$ are \emph{coupled} if $f(\xi_1,\xi_2)=0$ for some non-zero bivariate polynomial $f$ of degree at most $d_1$;
and \emph{non-coupled} otherwise.

\begin{lemma}\label{lemma:case-1-for-3-piers}
	For three input coordinates $\xi_1,\xi_2,\xi_3$ from the same group,
	there is a non-zero 3-variate polynomial $f$ with degree at most $d_1$
	such that $f(\xi_1,\xi_2,\xi_3)\equiv 0$.
\end{lemma}
}

%\begin{proof}[Proof of Lemma~\ref{lemma:case-1-for-3-piers}]
	Our proof makes use of terminologies about algebraic varieties and ideals (e.g.~\cite{IVA}). Denote by $\R[z_1,\ldots,z_d]$ the ring whose elements are sums of monomials in $z_1,\ldots,z_d$ with real coefficients.	
	
	By assumption of the lemma, $\xi_1$, $\xi_2$ and $\xi_3$ are in the group, say $G_k$, in $\mathcal{G}$, and they are governed by the hidden parameter $u_k$.   For $i \in [1,3]$, there exists a bivariate polynomial $g_i(u_k)$, where $g_i = h_{j_i,k}^x$ or $g_i = h_{j_i,k}^y$ for some $j_i \in [1,n]$, such that $\xi_i = g_i(u_k)$.
	
	Define three polynomials $f_1,f_2,f_3$ in $\R[p,q,r,s,t]$ as follows:
	\begin{align*}
		f_1(p,q,r,s,t) & = r - g_1(p,q), \\
		f_2(p,q,r,s,t) & = s - g_2(p,q), \\
		f_3(p,q,r,s,t) & = t - g_3(p,q).
	\end{align*}
	Let $\mathcal{J}$ be the ideal generated by $f_1$, $f_2$, and $f_3$.  That is,
	$$\mathcal{J} := \left\{\sum_{i=1}^3 \alpha_i f_i :  \alpha_1,\alpha_2,\alpha_3 \in \R[p,q,r,s,t] \right\}.$$  Let $\mathcal{B}$ be the unique reduced Gr\"{o}bner basis of $\mathcal{J}$ with respect to the lex order $p > q > r > s > t$. (The uniqueness is given in~\cite[p.~92,~Prop~6]{IVA}.  A basis of $\mathcal{J}$ is any set $f'_1,\ldots,f'_m$ of polynomials such that $\mathcal{J}=\{\sum_{i=1}^{m}\alpha_if'_i:\alpha_1,\ldots,\alpha_m\in \R[p,q,r,s,t]\}$, and a Gr\"{o}bner basis of $\mathcal{J}$ is a basis of $\mathcal{J}$ with a special property~\cite[p.~77,~Def~5]{IVA}, and a reduced Gr\"{o}bner basis requires further properties~\cite[p.~92, Def~5]{IVA}.)
	The degrees of $f_1$, $f_2$, and $f_3$ are at most $d_0$ because the degrees of $g_1$, $g_2$, and $g_3$ are at most $d_0$ by Definition~\ref{df:property}(a).
	%Moreover, $f_1$, $f_2$, and $f_3$ form a basis of $\mathcal{J}$.
	By the result of Dub\'{e}~\cite{GB-degree-bound}, the degrees of polynomials in $\mathcal{B}$ are bounded by $2((d_0^2/2) + d_0)^{16} = d_1$.
	
	Let $\mathcal{J}_2$ be the \emph{second elimination ideal} of $\mathcal{J}$ with respect to the same lex order $p > q > r > s > t$; formally, $\mathcal{J}_2=\mathcal{J} \cap \R[r,s,t]$.  Applying the Elimination Theorem~\cite[p.~116,~Thm~2]{IVA}, $\mathcal{B}_2 = \mathcal{B} \cap \R[r,s,t]$ is a Gr\"{o}bner basis of $\mathcal{J}_2$.
	
	Assume for the time being that $\mathcal{J}_2$ does not consist of the zero polynomial alone.
This implies the existence of a non-zero polynomial in $\mathcal{B}_2$, denoted by $f$.
The degree of $f$ is at most $d_1$ because $f \in \mathcal{B}_2\subseteq \mathcal{B}$ and the degree of every polynomial in $\mathcal{B}$ is at most $d_1$.

Denote $u_k=(x,y)$. Then, for $i \in [1,3]$, $\xi_i = g_i(x,y)$.
By the definitions of $f_1$, $f_2$, and $f_3$, it is obvious that for $i \in [1,3]$, $f_i(x,y,\xi_1,\xi_2,\xi_3) = 0$ for all $x,y \in \R$.
Because $f_1$, $f_2$, and $f_3$ form a basis of $\mathcal{J}$,
  $f(\xi_1,\xi_2,\xi_3)\equiv 0$ because $f(\xi_1,\xi_2,\xi_3) = \sum_{i=1}^3 \alpha_i(p,q,\xi_1,\xi_2,\xi_3) f_i(p,q,\xi_1,\xi_2,\xi_3) = 0$, thereby establishing the correctness of the lemma.
	
\smallskip What remains to be proved is that $\mathcal{J}_2$ does not consist of the zero polynomial alone.  Let $\CC$ be the set of complex numbers.  Let $\CC[z_1,\ldots,z_d]$ denote the ring whose elements are sums of monomials in $z_1,\ldots,z_d$ with complex coefficients.	Consider the ideal $\J$ generated by $f_1$, $f_2$ and $f_3$ in $\CC[p,q,r,s,t]$.  That is, $\J = \bigl\{\sum_{i=1}^3 \beta_if_i :  \beta_1,\beta_2,\beta_3 \in \CC[p,q,r,s,t] \bigr\}$.  Let $\J_2$ be the second elimination ideal of $\J$, i.e., $\J_2=\J\cap \CC[r,s,t]$.  It is known that if we compute the reduced Gr\"{o}bner basis of $\mathcal{J}$ using Buchberger's algorithm~\cite{IVA}, the result is also the reduced Gr\"{o}bner basis of $\J$.  Accordingly, $\mathcal{J}_2$ consists of the zero polynomial alone if and only if $\J_2$ consists of the zero polynomial alone.  Thus, it reduces to showing that $\J_2$ does not consist of the zero polynomial alone.
	
	Let $V(\J)=\{ (a,b,x,y,z) \in \CC^5 :   \forall \, f \in \J, \, f(a,b,x,y,z)=0 \}$,
	i.e., the subset of $\CC^5$ at which all polynomials in $\J$ vanish.
	Similarly, let $V(\J_2)= \{(x,y,z) \in \CC^3 : \forall \, f \in \J_2, \, f(x,y,z)=0 \}$.
	Define the projection $\varphi : \CC^5 \rightarrow \CC^3$ such that $\varphi(a,b,x,y,z) = (x,y,z)$.
	Then, $\varphi(V(\J))$ is the image of $V(\J)$ under $\varphi$.
	Let $\overline{\varphi(V(\J))}$ be the Zariski closure of $\varphi(V(\J))$, i.e., the smallest affine algebraic variety containing $\varphi(V(\J))$~\cite{IVA}.
	
	Because $\J$ is generated by $r - g_1(p,q)$, $s -g_2(p,q)$, and $t-g_3(p,q)$,
	we have $r=g_1(p,q)$, $s=g_2(p,q)$, and $t=g_3(p,q)$ for every element $(p,q,r,s,t) \in V(\J)$.
	In other words, once $p$ and $q$ are fixed, the values of $r$, $s$ and $t$ are completely determined.
	Hence, $V(\J)$ is isomorphic to $\CC^2$, which implies that $\dim(V(\J))=\dim (\CC^2)=2$.
	
	By the Closure Theorem~\cite[p.~125,~Thm~3]{IVA}, $V(\J_2)=\overline{\varphi(V(\J))}$.
	Therefore, $\dim(V(\J_2))=\dim(\overline{\varphi(V(\J))})$.
	It is also known that $\dim(\overline{\varphi(V(\J))}) \leq \dim(V(\J))$~\cite{IVA}.
	
	Altogether  $\dim(V(\J_2))\leq \dim(V(\J))=2$. Therefore, $V(\J_2)\neq \CC_3$  as $\dim(\CC_3)=3$.  This completes the proof because if $\J_2$ consists of the zero polynomial alone, $V(\J_2)$ would be equal to $\CC_3$.
	%It follows that $\J_2$ does not consist of the zero polynomial alone.
%\end{proof}

\subsection{Proof of Lemma~\ref{lemma:case-2-for-3-piers}}
\label{sec:case-2-for-3-piers}

\cancel{

\begin{lemma}\label{lemma:case-2-for-3-piers}
	For three non-constant and pairwise non-coupled input coordinates $\xi_1,\xi_2,\xi_3$ that are not in the same group,
	the distribution of $(\xi_1,\xi_2,\xi_3)$ is $d_1$-generic.
\end{lemma}

}

%\begin{proof}[Proof of Lemma~\ref{lemma:case-2-for-3-piers}]
	As $\xi_1$, $\xi_2$, and $\xi_3$ are not in the same group in $\mathcal{G}$, we can assume without loss of generality that $\xi_3$ is in a group in $\mathcal{G}$ different from those to which $\xi_1$ and $\xi_2$ belong.  By Definition~\ref{df:generic}, we need to prove that for every non-zero trivariate polynomial $f$ in $\xi_1$, $\xi_2$, and $\xi_3$ that has real coefficients and degree no more than~$d_1$, $\Pr \bigl[f({\xi}_1,\xi_2,\xi_3) = 0 \bigr]= 0$.
	
	Express $f $ as $f_0(\xi_1,\xi_2) \cdot \xi_3^{d_1}+f_1(\xi_1,\xi_2) \cdot \xi_3^{d_1-1}+\ldots+f_{d_1-1}(\xi_1,\xi_2) \cdot \xi_3+f_{d_1}(\xi_1,\xi_2)$, where $f_i(\xi_1,\xi_2) \in \R[\xi_1,\xi_2]$ and $f_i$ has degree at most $i$.  Since $f$ is a non-zero polynomial, there exists $j \in [0,d_1]$ such that $f_j$ is a non-zero polynomial and $f_i$ is the zero polynomial for $i \in [0,j-1]$.
	
	Denote by $R$ the set of roots of the equation $f_j(\xi_1,\xi_2) = 0$.  By the assumption of the lemma, $\xi_1$ and $\xi_2$ are uncoupled, which means that the distribution of $(\xi_1,\xi_2)$ is $d_1$-generic.  Thus, it follows from Definition~\ref{df:generic} that  $\Pr[ f_j(\xi_1,\xi_2) = 0] = 0$. In other words,
	\begin{equation}
		\Pr [ (\xi_1,\xi_2) \in R] = 0.  \label{eq:final-0}
	\end{equation}
	Fix an arbitrary point $(r,s) \in \R^2 \setminus R$. Then, $f(r,s,\xi_3)$ is a polynomial $g(\xi_3)$ in $\xi_3$ with degree at most $d_1$.   Moreover, $g(\xi_3)$ is not the zero polynomial because the coefficient $f_j(r,s)$ of the monomial $\xi_3^{d_1-j}$ is non-zero by our choice of $(r,s) \in \R^2 \setminus R$.  By the fundamental theorem of algebraic, there are at most $d_1$ roots to the equation $g(\xi_3) = 0$.   Since $\xi_3$ is not in the same group as $\xi_1$ or $\xi_2$, the distribution of $\xi_3$ is independent from that of $(\xi_1,\xi_2)$.  As $\xi_3$ is a non-constant input coordinate, the probability of $\xi_3$ being a root of $g(\xi_3)= 0$ is zero. Hence,
	\begin{equation}
		\forall \, (r,s) \in \R^2 \setminus R, \quad \Pr \bigl[ f(\xi_1,\xi_2,\xi_3) = 0 \, | \, \xi_1 = r, \xi_2 = s \bigr] = 0.  \label{eq:final-1}
	\end{equation}
	Let $p : \R^2 \rightarrow \R$ be the joint probability density function of the distribution of $(\xi_1,\xi_2)$.  We have
	\begin{eqnarray*}
		& & \Pr \bigl[f(\xi_1,\xi_2,\xi_3) = 0 \bigr] \\
		& = &
		\Pr \bigl[f(\xi_1,\xi_2,\xi_3) = 0 \, \wedge \, (\xi_1,\xi_2) \not\in R\bigr] + \Pr \bigl[f(\xi_1,\xi_2, \xi_3) = 0 \, \wedge \, (\xi_1,
		\xi_2) \in R\bigr] \\
		& \stackrel{\eqref{eq:final-0}}{=} & \Pr \bigl[f(\xi_1,\xi_2,\xi_3) = 0 \, \wedge \, (\xi_1,\xi_2) \not\in R\bigr] \\
		& = & \iint_{\R^2 \setminus R} p(r,s) \cdot \Pr \bigl[ f(\xi_1,\xi_2,\xi_3) = 0 \, | \, \xi_1 = r, \xi_2 = s\bigr] \, \mathrm{d}r \, \mathrm{d}s \\
		& \stackrel{\eqref{eq:final-1}}{=} & 0.
	\end{eqnarray*}
	Hence, the distribution of $(\xi_1,\xi_2,\xi_3)$ is $d_1$-generic by Definition~\ref{df:generic}, establishing the lemma.
%\end{proof}

\subsection{Proof of Lemma~\ref{lemma:learn-partition-DT-by-linearity}}
\label{sec:learn-partition-DT-by-linearity}

We first define an extension of a vector in $\R^m$ as follows.  Given a positive integer $d$ and a vector $(r_1,\ldots,r_m)\in \R^m$, we can extend the vector to a longer one that consists of all possible monomials in $r_1,\ldots,r_m$ whose degrees are at most $d$.  Let $\kappa = {{m+d}\choose{m}}$.  There are $\kappa$ such monomials, and we list them out in lexicographical order, i.e., $r_1^{d_1} \cdots r_m^{d_m} < r_1^{d'_1} \cdots r_m^{d'_m}$ if and only if there exists $j \in [1,m]$ such that for $d_i = d'_i$ for $i \in [1,j-1]$ and $d_j < d'_j$.  We use ${\mathcal{E}}_d(r_1,\cdots,r_m)$ to denote the extended vector of $(r_1,\ldots,r_m)$ as defined above.
	
Let $\xi_1,\ldots,\xi_m$ be the input coordinates that we are interested in.  Draw a sample of $\kappa$ input instances.  For $i \in [1,\kappa]$, let $\bigl(\xi_1^{(i)},\ldots, \xi_m^{(i)}\bigr)$ denote the instance of $(\xi_1,\ldots,\xi_m)$ in the $i$-th input instance drawn.  For $i \in [1,\kappa]$, let $\pmb{q}_i = {\mathcal{E}}_{d}\bigl(\xi_1^{(i)}, \ldots, \xi_m^{(i)}\bigr)$.  We show that we can use $\pmb{q}_1,\ldots,\pmb{q}_\kappa$ to decide whether condition~(i) or~(ii) below is satisfied by $\xi_1,\xi_2,\ldots,\xi_m$, assuming that exactly one of conditions~(i) and~(ii) is satisfied.
\begin{itemize}
\item Condition~(i): There exists a non-zero $m$-variate polynomial $f$ of degree at most $d$ such that $f(\xi_1,\ldots,\xi_m)\equiv 0$.
\item Condition~(ii): The distribution of $(\xi_1,\ldots,\xi_m)$ is $d$-generic.
\end{itemize}

Our method is to test the linear dependence of $\pmb{q}_1,\ldots,\pmb{q}_{\kappa}$ by running Gaussian elimination on the $\kappa \times \kappa$ matrix with columns equal to $\pmb{q}_1,\ldots,\pmb{q}_{\kappa}$.  This takes $O(\kappa^3)$ time.  If $\pmb{q}_1,\ldots,\pmb{q}_{\kappa}$ are found to be linear dependent, we report that condition~(i) is satisfied.  Otherwise, we report that condition~(ii) is satisfied.
To show that our answer is correct almost surely, it suffices to prove the following statements:
\begin{itemize}
		\item If condition~(i) holds, then $\pmb{q}_1,\ldots,\pmb{q}_{\kappa}$ are linearly dependent.
		\item If condition~(ii) holds, then $\pmb{q}_1,\ldots,\pmb{q}_{\kappa}$ are linearly independent almost surely.
\end{itemize}
	
Suppose that condition~(i) holds. Then, there is a non-zero $m$-variate polynomial $f$ of degree $d$ such that $f(\xi_1,\ldots,\xi_m)\equiv 0$.  Observe that $f(\xi_1,\ldots,\xi_m)$ is equal to the inner product $\langle {\mathcal{E}}_{d}(\xi_1,\ldots,\xi_m),\pmb{a} \rangle$ for some non-zero vector $\pmb{a}\in \R^\kappa$.  It follows that for $i \in [1,\kappa]$, $\langle \pmb{q}_i,\pmb{a}\rangle = f(\xi_1^{(i)},\ldots,\xi_m^{(i)}) = 0$.  In other words, the vectors $\pmb{q}_1,\ldots,\pmb{q}_\kappa$ in $\R^{\kappa}$ are all orthogonal to $\pmb{a}$.  Since the orthogonal complement of $\pmb{a}$ has dimension $\kappa-1$, the vectors $\pmb{q}_1, \ldots, \pmb{q}_{\kappa}$ must be linearly dependent.
	
Suppose that condition~(ii) holds.  We prove by induction that $\pmb{q}_{1},\ldots,\pmb{q}_{j}$ are linearly independent almost surely for $j = 1, 2, \ldots,\kappa$.

Consider the base case of $j = 1$.  The vector $\pmb{q}_1$ is linearly independent because $\pmb{q}_1 \not= (0,\ldots,0)$.
    This is because the first dimension of $\pmb{q}_1$ is always $1$.
	
Assume the induction hypothesis for some $j=k < \kappa$.  Consider $j=k+1$.  Since $k < \kappa$, the $k$ vectors $\pmb{q}_1, \ldots, \pmb{q}_k$ cannot span $\R^{\kappa}$.  Therefore, there exists a non-zero vector $\pmb{a}$ that is orthogonal to $\pmb{q}_i$ for all $i \in [1,k]$.  Consider the equation $\langle {\mathcal{E}}_{d}\bigl(\xi_1, \ldots, \xi_m\bigr)\cdot \pmb{a} \rangle =0$.  We can write this equation as $f\bigl({\xi}_1,\ldots,\xi_m\bigr) = 0$, where $f$ is a sum of monomials in ${\xi}_1,\ldots,\xi_m$ with coefficients equal to the corresponding entries in $\pmb{a}$.  Since $\pmb{a}$ is non-zero vector, $f$ is a non-zero $m$-variate polynomial of degree at most $d$.  By condition~(ii), the distribution of $(\xi_1,\ldots,\xi_m)$ is $d$-generic, which gives $\Pr \bigl[f({\xi}_1,\ldots,\xi_m) = 0 \bigr]= 0$ according to Definition~\ref{df:generic}.  Therefore, $\Pr[\langle \pmb{q}_{k+1}, \pmb{a} \rangle = 0] = \Pr \bigl[f\bigl({\xi}^{(k+1)}_1,\ldots,\xi^{(k+1)}_m\bigr) = 0 \bigr] = 0$.
	
By our choice of $\pmb{a}$, we have $\langle \pmb{q}_i, \pmb{a} \rangle = 0$ for $i \in [1,k]$.  It implies that $\langle \pmb{q}_{k+1}, \pmb{a} \rangle = 0$  if $\pmb{q}_{k+1}$ is equal to some linear combination of $\pmb{q}_1,\ldots,\pmb{q}_k$.  Therefore, $\pmb{q}_{k+1}$ is not a linear combination of $\pmb{q}_1,\ldots,\pmb{q}_{k}$ almost surely because $\Pr[\langle \pmb{q}_{k+1}, \pmb{a} \rangle = 0] = 0$.  Together with the induction hypothesis that $\pmb{q}_1,\ldots,\pmb{q}_k$ are linear independent almost surely, we conclude that $\pmb{q}_1,\ldots,\pmb{q}_{k+1}$ are linearly independent almost surely.

%\end{proof}

\end{document}